\documentclass[11pt]{article}

\usepackage{authblk}
\usepackage{subcaption}
\usepackage{graphicx}
\usepackage{wrapfig}
\usepackage[dvipsnames]{xcolor}
\usepackage{graphics}
\usepackage{nicefrac}
\usepackage{enumitem}
\usepackage{latexsym}
\usepackage{theorem}    
\usepackage{amsfonts}
\usepackage{amssymb}
\usepackage{hyperref}
\usepackage{xspace}
\usepackage{latexsym}
\usepackage{amsmath}
\usepackage{pgf,tikz,pgfplots,mathrsfs}
\usepackage{algorithm}
\usepackage[noend]{algpseudocode}
\newcommand{\NoD}{{\sf No-Detour }}
\newcommand{\OneD}{{\sf 1-Detour }}
\newcommand{\TwoD}{{\sf 2-Detour }}
\newcommand{\XthreeC}{{\sf X3C }}
\newcommand{\XoneC}{{\sf X1C }}
\newcommand{\CXP}{{\sf CXP }}

\newcounter{trajectory}
\newenvironment{trajectory}
    {
    \refstepcounter{trajectory}
    \begin{center}
    \begin{tabular}{|p{0.9\textwidth}|}
    \hline
    \textbf{Trajectories \thetrajectory.}
    }
    { 
    \\\hline
    \end{tabular} 
    \end{center}
    }
    
\usepackage{todonotes}

\newtheorem{theorem}{Theorem}[section]

\newtheorem{lemma}{Lemma}[section]
\newtheorem{obs}{Observation}

\newtheorem{remark}{Remark}[section]

\begin{document}
\newenvironment{proof}{\noindent{\bf Proof.} }{\null\hfill$\Box$\par\medskip}
\newenvironment{proofi}{\par{\it Proof}. \ignorespaces}{}

\title{Evacuation of equilateral triangles by mobile agents \\of limited 
  communication range\footnote{This research was
    supported by NSERC of Canada}}
\author{Iman Bagheri, Lata  Narayanan, Jaroslav Opatrny\\
Department of Computer Science and Software Engineering,\\
  Concordia University, Montreal, Canada\\
       {\textit {\{imanbag@gmail.com,
lata@encs.concordia.ca,opatrny@cs.concordia.ca\}}}}
\maketitle

\maketitle

\begin{abstract}
  We consider the problem of evacuating $k \geq 2$ mobile agents from  a unit-sided equilateral triangle through an exit located at an unknown location on the perimeter of the triangle. The agents are initially located at the centroid of the triangle and they can communicate with other agents at distance at most $r$ with $0\leq r \leq 1$.  An agent can move at speed at most one, and  finds the exit only when it reaches the point where the exit is located. The agents can collaborate in the search for the exit. The goal of the {\em evacuation problem} is to minimize the evacuation time, defined as the worst-case time for {\em all} the agents to reach the exit.
  
  We propose and analyze several algorithms for the problem of evacuation by $k \geq 2$   agents; our results indicate that the best strategy to be used varies depending on the values of  $r$ and $k$. For two agents, we give three algorithms, each of which achieves the best performance for different sub-ranges of $r$ in  the range $0 \leq r \leq 1$. 
We also show a lower bound on the evacuation time of two agents for any $r < 0.336$. For $k >2$ agents, we study three strategies for evacuation: in the first strategy, called {\sf X3C}, agents explore all three sides of the triangle before connecting to exchange information; in the second strategy,  called {\sf X1C}, agents explore a single side of the triangle before connecting; in the third strategy, called {\sf CXP}, the agents travel to the perimeter to locations in which they are connected, and explore it while always staying connected.  For 3 or 4 agents, we show that \XthreeC works better than \XoneC for small values of $r$, while \XoneC works better for larger values of $r$. Finally, we show that for any $r$, evacuation of $k=6 +2\lceil(\frac{1}{r}-1)\rceil$ agents can be done using the \CXP strategy  in time $1+\sqrt{3}/3$, which is optimal in terms of time, and asymptotically optimal in terms of the number of agents. 
\end{abstract}


\section{Introduction}
Consider the situation where several mobile agents/robots are located inside a closed region, that has a single exit point on the perimeter of the region at a location unknown to the agents. Due to some emergency, the agents all need to leave this region as quickly as possibly. Thus the agents need to collaboratively search for the exit and  minimize the time that is needed for {\em all of them} to reach the exit. This  {\em evacuation problem} has already been considered
for several different regions and agents of different capabilities.

Two models of communication between the agents have been  considered  in the context of group search and evacuation. 
In the first model, called the {\em face-to-face} model, the agents can communicate only when they are in the same place at the same time. 
In the second model, called the {\em wireless} model, the agents can communicate at any time and over any distance. The algorithms for evacuation in the wireless and face-to-face models are in general, quite different. It is natural to ask how the agents would perform evacuation if their communication range was limited to some  $r$ with $0 \leq  r \leq \infty$ where the diameter of the region is assumed to be 1. Note that $r=0$ gives the face-to-face model, and $r = \infty$ corresponds to the wireless model. In any
regions of diameter 1, since agents never need to go outside the region to communicate with other agents, a communication range of $r >1$ confers no advantages. Thus the wireless model is equivalent to the case $r=1$.

In this paper we study the problem of evacuating
an equilateral triangle whose sides are of size $1$ with agents located initially in the centroid of the triangle,
and whose communication capabilities are limited to a given, fixed  distance $0< r<1$. To the best of our knowledge, the case of  limited range communication of agents in evacuation problems has not been considered yet.
Since the evacuation of the equilateral triangle was previously studied for
the face-to-face model \cite{ChuangpishitMNO17} and the wireless model \cite{EtrSq}, it will allow us to evaluate the impact of the limited transmission range on the evacuation algorithms.  
When there are three or more agents, then an agent can act as a  {\em relay} between two other agents, thereby increasing the {\em effective communication range} of the agents. Indeed a {\em virtual meeting} can occur between many agents, even when they are not co-located, so long as the network of communication they create is connected. This suggests that the  interplay between the communication range $r$ and the number of agents $k$ should be considered in the design of evacuation algorithms.

\subsection{Our Results}
We first study in detail the case of $k=2$ agents. In Section \ref{sec:2a}  we propose three evacuation algorithms for two agents, parametrized by $r$, and we establish the ranges of $r$ for which each gives the 
best results. As shown in Table ~\ref{table-summary}, throughout the entire range $0 \leq r \leq 1$, our algorithms take advantage of increased communication range to achieve lower evacuation time. 
In \cite{ChuangpishitMNO17}, it was shown that for $r=0$, it helps for the agents to make {\em detours} into the interior of the triangle, and in fact, more detours are shown to always improve the evacuation time, though it should be noted that the improvement for more than 2 detours is negligible. We show here that for $r > 0.7375$, using {\em any} detour worsens the evacuation time, while using more than one detour is not useful for $r > 0.4725$. We also show a lower bound of $1+2/\sqrt{3} -r$ on the evacuation time of two agents for any $r<0.366$. 

For $k > 2$ agents, we investigate three different strategies for evacuation. In the first strategy, called {\em Explore 3 sides before Connecting }({\sf X3C}), the perimeter of the triangle is partitioned into $k+1$ segments. 
The agents move to explore $k$ segments on all three sides, subsequently entering the interior of the triangle to form a connected network in order to communicate the results to the other agents, after which they either move to the exit or they all 
explore the remaining segment. In the second strategy, called {\em Explore 1 Side before Connecting } ({\sf X1C}) only one of the sides of the triangle is partitioned  into multiple segments, each to be explored by an agent. 
At the end of the exploration of the edge, two of the agents explore the remaining two sides of the triangle,  
while the other agents move inside to create and maintain connectivity of all agents. As soon as the exit is found, all agents can move to the exit. In the final strategy (which is only possible if the number of agents is large enough relative to $r$), called {\em Connected Exploration of Perimeter }({\sf CXP}), the agents move to positions  over two sides of the perimeter to ensure that the agents are connected {\em before} they start exploration, and they stay connected during the entire exploration.

We study in detail the case of 3 and 4 agents in Section \ref{sec:3a}. Note that the \CXP   strategy cannot apply in these cases, and thus we study only the {\sf X3C} and {\sf X1C} strategies. 
Our results show that \XthreeC works better than \XoneC for smaller values of $r$ and \XoneC is better for larger values of $r$; see Table ~\ref{table-summary}. 

Finally we consider in Section \ref{sec:ka} the problem of the optimal evacuation of $k$ agents. It was shown in \cite{EtrSq} that for any $r$, regardless of the number of agents, evacuation cannot be done in time less that 
$1 + \sqrt{3}/3$; on the other hand, this time can be achieved by 6 agents and $r=1$.  In this paper we show that for any $r >  0$, 
evacuation can achieved in the optimal time of $1+\sqrt{3}/3$ if the number of agents is $6 +2\lceil(\frac{1}{r}-1)\rceil$. Indeed for $r=1/2$, eight agents suffice, and for $r=1/3$, ten agents suffice, for $r=1/4$, twelve agents suffice. We also show that $\Omega(1/r)$ agents are required to evacuate in time $1+ \sqrt{3}/3$.

We conjecture that for any $k \geq 6$ agents, there exist $r_1, r_2$ with $0 < r_1 < r_2 < 1$ such that \XthreeC is the best strategy of the three for $0 \leq r \leq r_1$, \XoneC is the best strategy for $r_1 < r \leq r_2$, and 
\CXP is the best strategy for $r_2 < r \leq 1$.

\begin{table}
  \begin{center}
\begin{tabular}{|c||c|c||c|c||c|c|}\hline
  &\multicolumn{2}{c||}{Two Agents}
  &\multicolumn{2}{c||}{Three Agents}
  &\multicolumn{2}{c|}{Four Agents}\\ \hline
 r&$\ $ evac. time$\ $ &$\ $ alg.$\ $ &$\ $ evac. time$\ $ &$\ $ alg. $\ $&$\ $ evac. time$\ $ &$\ $ alg. $\ $\\ \hline
0& 2.3367& see \cite{ChuangpishitMNO17}&2.0887& see \cite{ChuangpishitMNO17}
 & 1.98157& see \cite{ChuangpishitMNO17}\\ \hline
 0.1&2.25424&\TwoD&2.08871& \XthreeC & 1.96199& \XthreeC\\ \hline
 0.2&2.18584&\TwoD&2.07642&\XthreeC & 1.88392& \XoneC\\ \hline
 0.3&2.12325&\TwoD&1.93620&\XoneC &1.67649 & \XoneC\\ \hline
 0.4&2.06506&\TwoD&1.78880&\XoneC &1.62573 & \XoneC\\ \hline
 0.5&2.01050&\OneD&1.68958&\XoneC & 1.61912 & \XoneC \\ \hline
 0.6&1.95926&\OneD&1.67532&\XoneC & 1.61302 & \XoneC \\ \hline
 0.7&1.91169&\OneD&1.66666&\XoneC & 1.61050 & \XoneC \\ \hline
 0.8&1.86559&\NoD&1.66666&\XoneC & 1.61050 & \XoneC \\ \hline
  0.9&1.82439&\NoD&1.66666&\XoneC & 1.61050 & \XoneC \\ \hline
1& 1.78867& see \cite{CKKNOS} & 1.66666 &see \cite{EtrSq} &1.61050 &see \cite{EtrSq} \\ \hline 
\end{tabular}
  \end{center}
  \caption{A summary of the evacuation times of our algorithms.}
  \vspace*{-10mm}
 \label{table-summary}
  \end{table}

 \subsection{Related work}

The evacuation problem was introduced in  \cite{CGGKMP} for agents inside a disk in both the wireless and face-to-face communication models. The authors gave optimal algorithms for 2 agents in the wireless model, and proved upper and lower bounds for the evacuation time for 2 agents in the face-to-face model. They also considered the problem for 3 agents and showed asymptotically tight bounds for $k$ robots in both models. The problem for the face-to-face model  was revisited in
\cite{CGKNOV}, and the results further improved in \cite{zurich}.
The evacuation of an equilateral triangle with agents in the wireless model  was considered in \cite{EtrSq}, and in the face-to-face
communication model in \cite{ChuangpishitMNO17}.  
  We should also mention the work on polygons \cite{FGK10}, evacuation of circle with faulty agents \cite{CGKNOV}, and
 the case of multiple exits on a circle \cite{DBLP:conf/icdcn/CzyzowiczDGKM16,pattanayak2017evacuating}.

 The evacuation problem is related to many other problems that have been considered previously. It can be seen as a variation of a search problem. In this context we should mention the classical {\em cow-path} problem, i.e., a problem of searching on a line  \cite{baezayates1993searching,B64,beck1970yet}, several of its versions  \cite{demaine2006online,kao1996searching,koutsoupias1996searching},
 a  group search on a line \cite{Groupsearch}, and a search on a line with
 faulty agents \cite{PODC16}.
 There are many studies involving mobile, autonomous agents  in the plane \cite{FPS12}.
 The problem of search \cite{EmekLSUW15,Brandt2018}, gathering of agents \cite{agmon2006fault,dieudonne2014gathering} in the plane, pattern formation \cite{FPSW2008}, etc.,
 have been done. The cop-and robber games \cite{bookBN}, and graph searches \cite{Fraigniaud2005} are also related.

\section{Model and Notation}
\begin{figure}
	\centering
	\includegraphics[width=6cm]{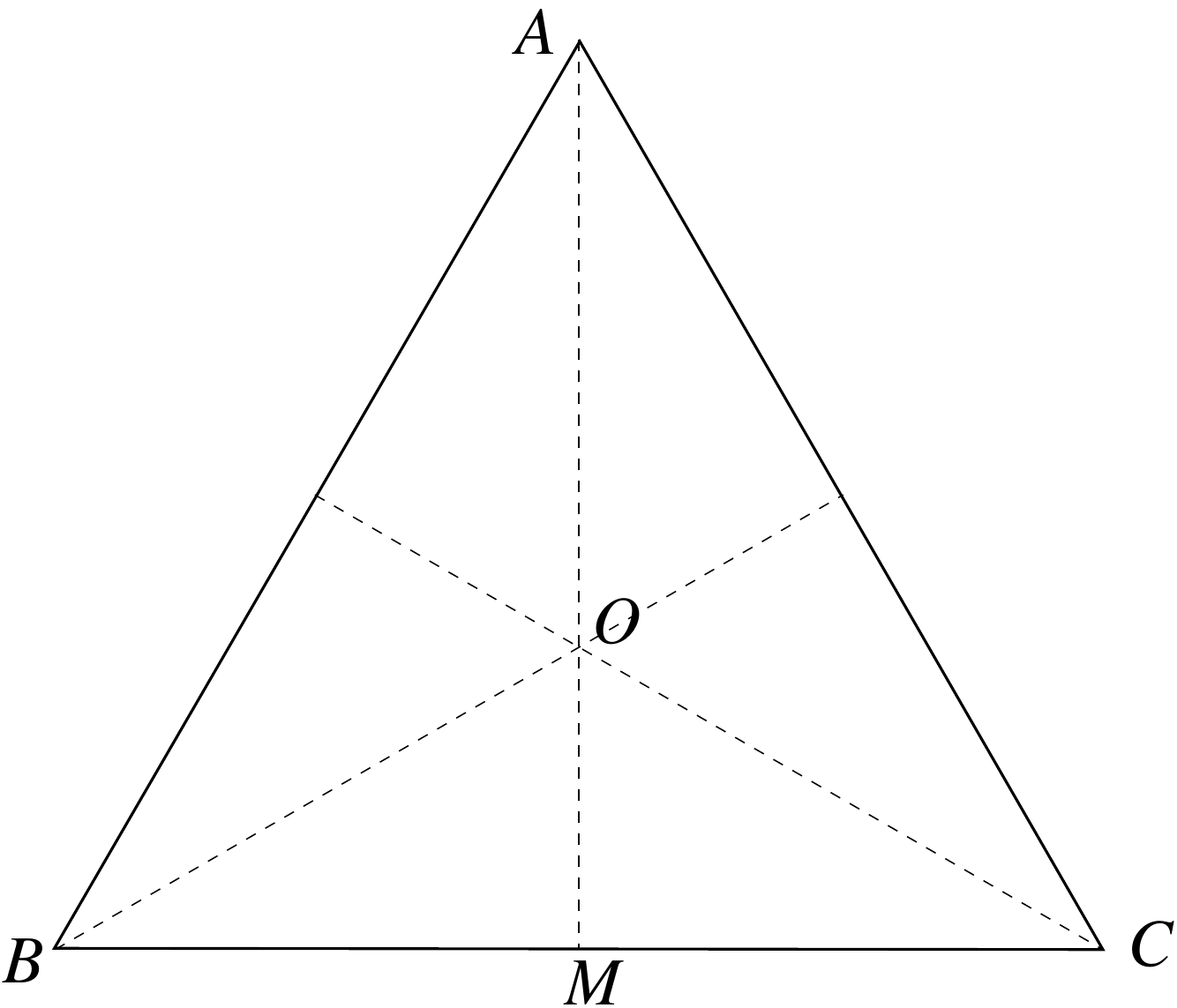}
	\caption{Equilateral triangle $T$.}
	\label{fig:triangle}
        \end{figure}
The search domain considered in this paper is the perimeter of an equilateral triangle with side 1.
We denote the triangle by $T$, with  vertices $A$, $B$ and $C$ starting at the top of the triangle, going counter-clockwise,  and  the centroid of the triangle by $O$, as  in  Figure
\ref{fig:triangle}.
Point $M$ is the midpoint of the segment $BC$. The \textit{height} of the triangle is denoted by $h$ and $y=h/3$.
The line segment connecting any two points $P$ and $Q$ is denoted by $PQ$ and its length by $|PQ|$. 
  Agents are initially located at the centroid $O$ of the triangle. 
Each agent can move at speed at most $1$, and it has a wireless transmitter/receiver with range $r\leq 1$. Unless specified otherwise in the algorithm, agents always move with speed 1.
 Agents are able to 
 carry out simple computations, e.g., if an agent finds the exit it can calculate the path to follow in order to inform other agents about the exit.
 In this paper agents are assumed to be non-faulty, meaning that they:
 follow their assigned trajectory, 
 recognize the exit if they reach its location, 
 and they can always exchange information if their distance is less than or equal to $r$.

Each agent follows a path, called its {\em trajectory}, assigned to it before the exploration begins.
We specify each evacuation algorithm by specifying a trajectory of each agent and its actions. 
An agent may leave its predetermined trajectory only if 
either it has found the exit point, or it has been notified by another agent about the location of the exit.
For each of these two situations the algorithm specifies the action to be followed.
e.g., when an agent finds the exit it specifies a point to go to from which
it can notify other agents about the exit, or if it has been notified about the exit location the action is  ``go to that point''. 
	
We denote the time that point $x$ is seen for the first time by either of the agents by $t_x$.
By $E_{\mathcal{A}}(k,r)$ we mean the worst-case evacuation time of algorithm $\mathcal{A}$ with $k$ agents, $k\geq2$ and communication range of $r$, $0\leq r \leq 1$. We denote the optimal evacuation time by $k$ agents by $E^*(k,r)$, that is:
$$E^*(k,r) = min_{\mathcal{A}}E_{\mathcal{A}}(k,r)$$

We define an {\em $r$-interception} to be the action of moving to a point in which the agent is at distance at most $r$  of the other agent(s) so that it can transmit the location of the exit point to them.
In all our algorithms, the trajectory of each agent is a sequence of line segments. To analyze the algorithms, we identify on each segment a {\em critical point},
defined to be the point or the immediate neighbourhood of a point where
the evacuation time on the segment is maximized. In order to minimize the maximum evacuation time, after identifying these critical points,  we optimize the algorithms by adjusting some parameters in the trajectories.

\section{Evacuation of Two Agents}\label{sec:2a}

In this section we give upper and lower bounds on the evacuation time for two agents with $0<r<1$.  Recall that the best known algorithm described in 
\cite{ChuangpishitMNO17} for the face-to-face model ($r=0$), evacuates two agents in time $2.3367$ and employs two detours per agent. 
In \cite{CKKNOS}, an optimal algorithm for the wireless model ($r=1$) with evacuation time of $\nicefrac{3}{2}+y\approx 1.78867$ is described. Hence if the agents are capable of communication within a certain range $0<r<1$, it is clear that the evacuation time should lie between these two values. 
We divide the triangle into two halves by a vertical line through $A$ and $O$, as shown in Figure \ref{fig:2r-nodetour-traj}. The trajectories of the two agents presented in this section are {\em symmetric} with respect to line $AO$. Thus the trajectory of the first agent $R_1$ includes exploration of the left half of the perimeter, and the second agent $R_2$ is responsible for exploration of the right half of $T$.  Therefore, without loss of generality, in the analysis of algorithms we will assume that the exit is located in the right half of the triangle throughout this section. 

The three evacuation algorithms for two agents presented in  this section use the same generic Algorithm \ref{alg:2r-nodetour-traj} given below. They only differ in the trajectories of the agents.   
\begin{algorithm}[h]
	\caption{Generic 2-agent Evacuation Algorithm Followed by an Agent.}\label{alg:2r}
		\begin{algorithmic}
			\Function{Exploration}{}
				\State found$\gets$ false
				\While{not\textless found\textgreater \, and not\textless msg\_recd \textgreater}
					\State move along the predetermined trajectory
				\EndWhile				
				\State \Call{action}{}
			\EndFunction
		\end{algorithmic}
		\begin{algorithmic}
			\Function{Action}{}
				\If{found}
					\State P $\gets$ current location
					\If{the other agent is not within communication range}
						\State calculate the closest point $U$, where the other agent can be 
						$r$-intercepted
						\State go to $U$
					\EndIf
					\State $send(P)$ to the other agent
				\EndIf
                                \State go to $P$ and exit
			\EndFunction
		\end{algorithmic}
		\label{alg:2r-nodetour-traj}
\end{algorithm}

Let $S_1$ and $S_2$ be points on the sides $AB$ and $AC$ at distance $r$ from $A$, shown in Figure \ref{fig:2r-nodetour-traj}. Points $S_1,S_2,A$ form an equilateral triangle 
at the top of $T$ with side $r$.
If the agents do not find the exit outside $\Delta S_1AS_2$ and enter this smaller 
triangle, they are always within communication range with each other, and the evacuation time for the three algorithms described in this section, is independent of the exit position and will always be $t_{S_1}+r$.

\begin{figure}[h]
  \centering
  \includegraphics[scale=0.3]{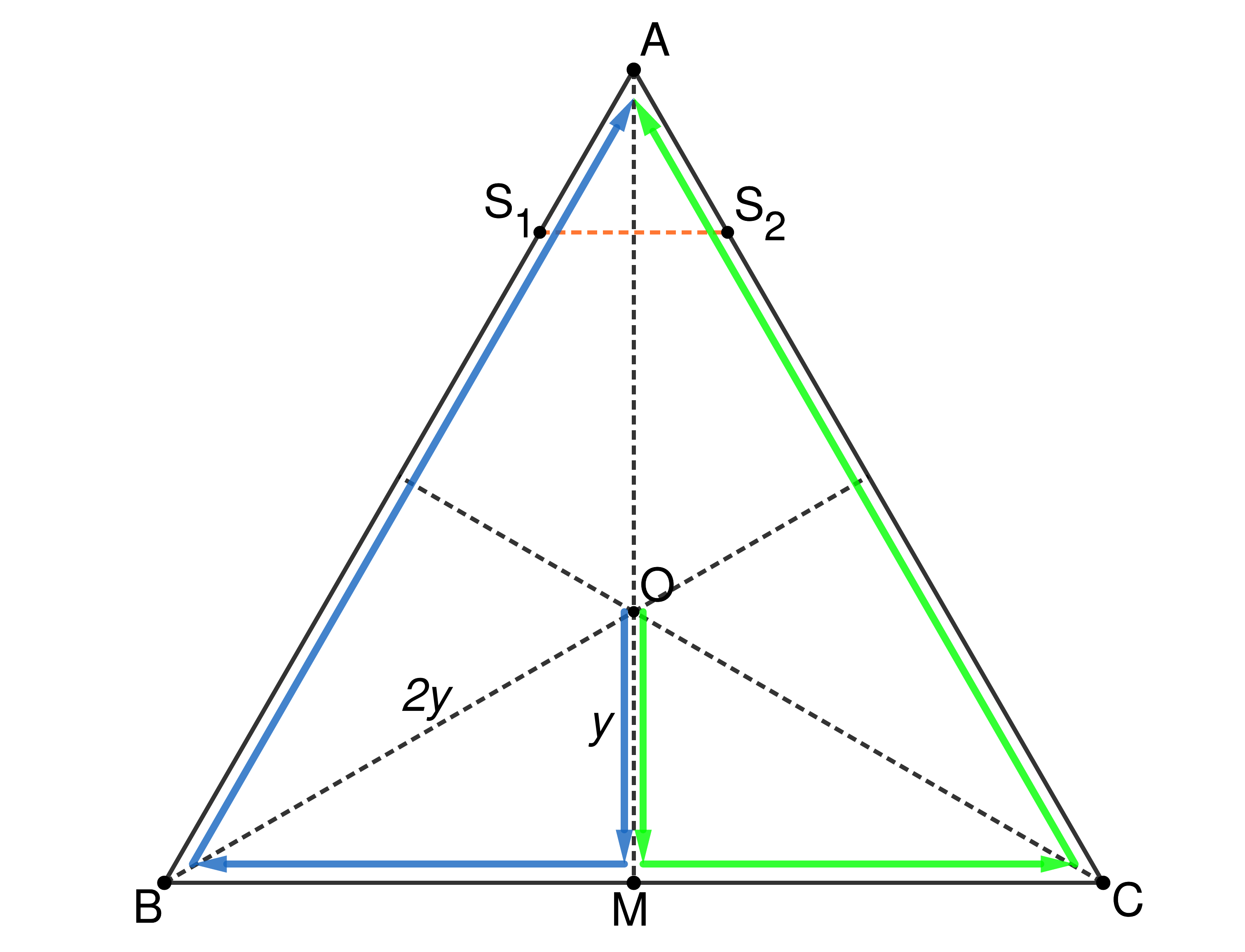}
 \caption{\NoD algorithm}
 \label{fig:2r-nodetour-traj}
 \end{figure}

\subsection{The \NoD Algorithm}

The trajectories of both agents are shown in Figure \ref{fig:2r-nodetour-traj} and defined in Trajectories \ref{traj:2r1}.
The trajectory of $R_1$ is shown in blue, 
the green trajectory is for $R_2$. Clearly, these trajectories do the fastest
possible exploration of the perimeter of $T$, and these trajectories are
known to give the  optimal time of $y+1.5$ for the wireless evacuation of $T$ by two agents starting in $O$. 
 \begin{trajectory} \label{traj:2r1}
$\ $ \NoD \\
\hspace*{2cm} $R_1\mbox{ follows the trajectory}:<O,M,B,A>$\\
\hspace*{2cm}  $R_2\mbox{ follows the trajectory}:<O,M,C,A>$ 
\end{trajectory}

The Algorithm \NoD uses the generic Algorithm \ref{alg:2r} with respect to Trajectories \ref{traj:2r1}. For the analysis of this algorithm, we assume the exit is found by $R_2$. Then we show that the maximum evacuation time is when the exit is located at point $C$.

In order to determine the critical point of some segments in $T$ we use the following lemma, which is a simple generalization of Theorem 1 in \cite{zurich}
for the case $r>0$.

\begin{lemma} \label{lem:betagamma}
  \cite{zurich} Suppose $R_1$ and $R_2$ with $r>0$ are looking for an exit on lines $L_1$ and $L_2$ respectively, as on Figure \ref{fig:betagamma}. Assume the exit is found by $R_2$ at point $N$, and  $Q$ be the point where $R_1$ is r-intercepted. Let $S$ be the line connecting $N$ and $Q$, $\beta$ be  the angle between $L_2$ and $S$, and $\gamma$ be the angle between $L_1$ and $S$ by $\gamma$. 

If $2\cos\beta+\cos\gamma < 1$ then
shifting the exit in the direction of the movement of $R_2$ yields a larger evacuation time, while  if
		$2\cos\beta+\cos\gamma>1$, then shifting the exit in the opposite direction of the movement of $R_2$ yields a larger evacuation time.
\end{lemma}

 \begin{figure}[ht]
	\centering
	 \includegraphics[scale=0.6]{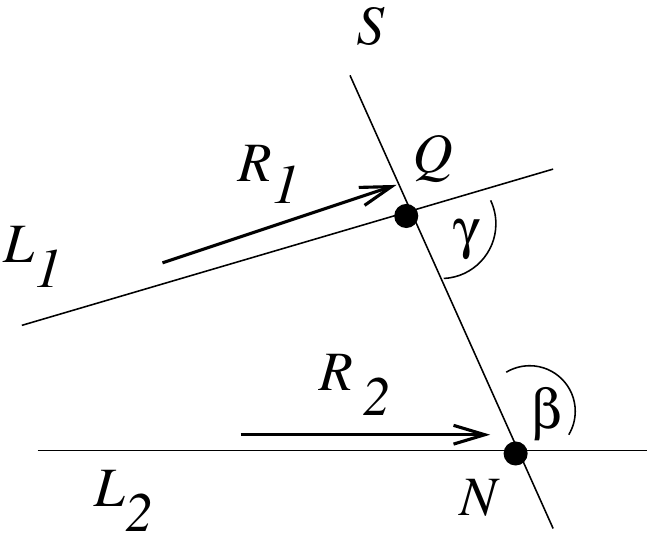}
	\caption{Illustration for Lemma \ref{lem:betagamma}.}
	\label{fig:betagamma}
\end{figure}

 \begin{proof} See \cite{zurich} \end{proof}

   \begin{figure}[!h]
\begin{subfigure}{.5\textwidth}
  \centering
  \includegraphics[scale=0.22]{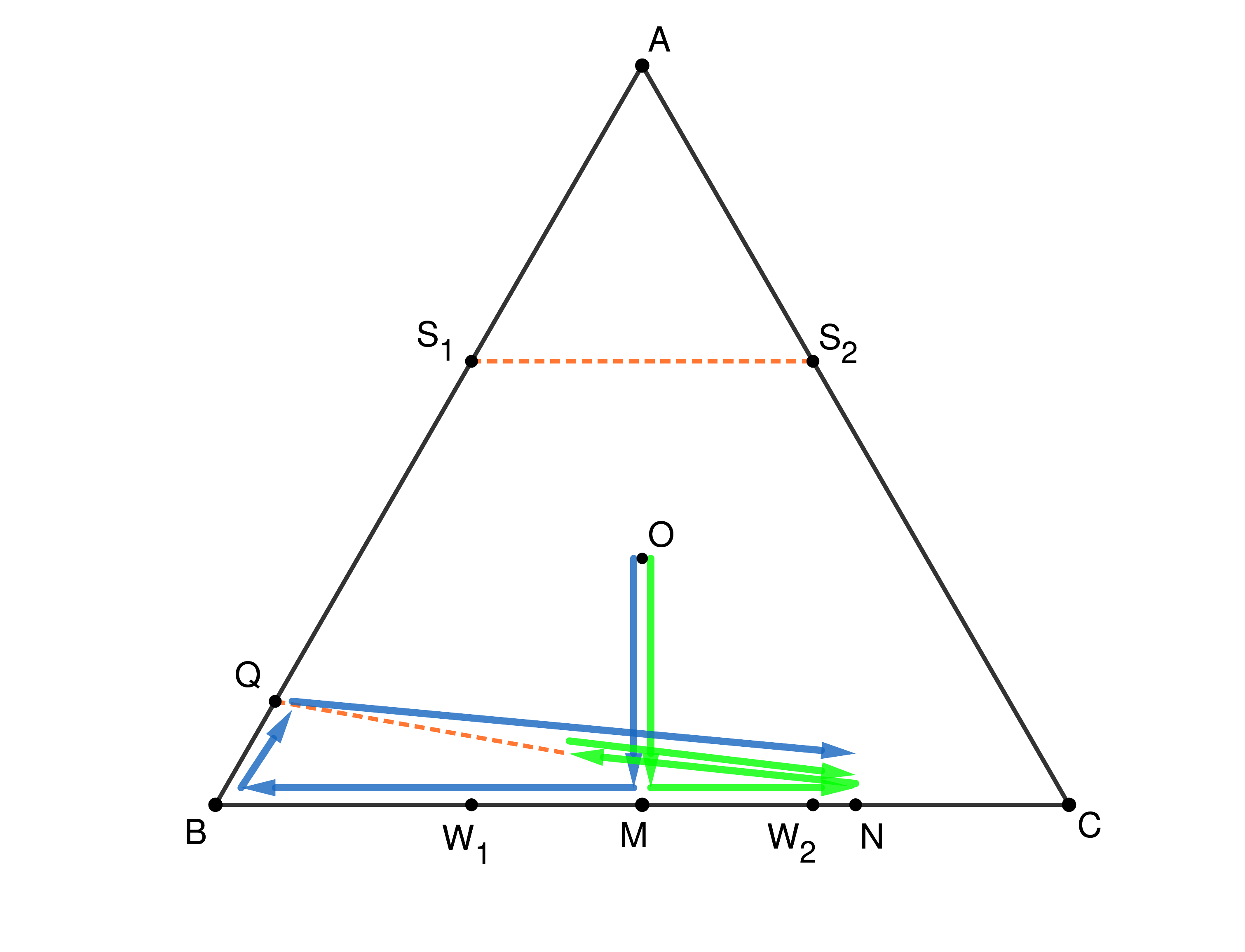}
  \caption{Exit is located on segment $MC$}
  \label{subfig:2r-nodetour-t2}
\end{subfigure}
\begin{subfigure}{.5\textwidth}
  \centering
  \includegraphics[scale=0.22]{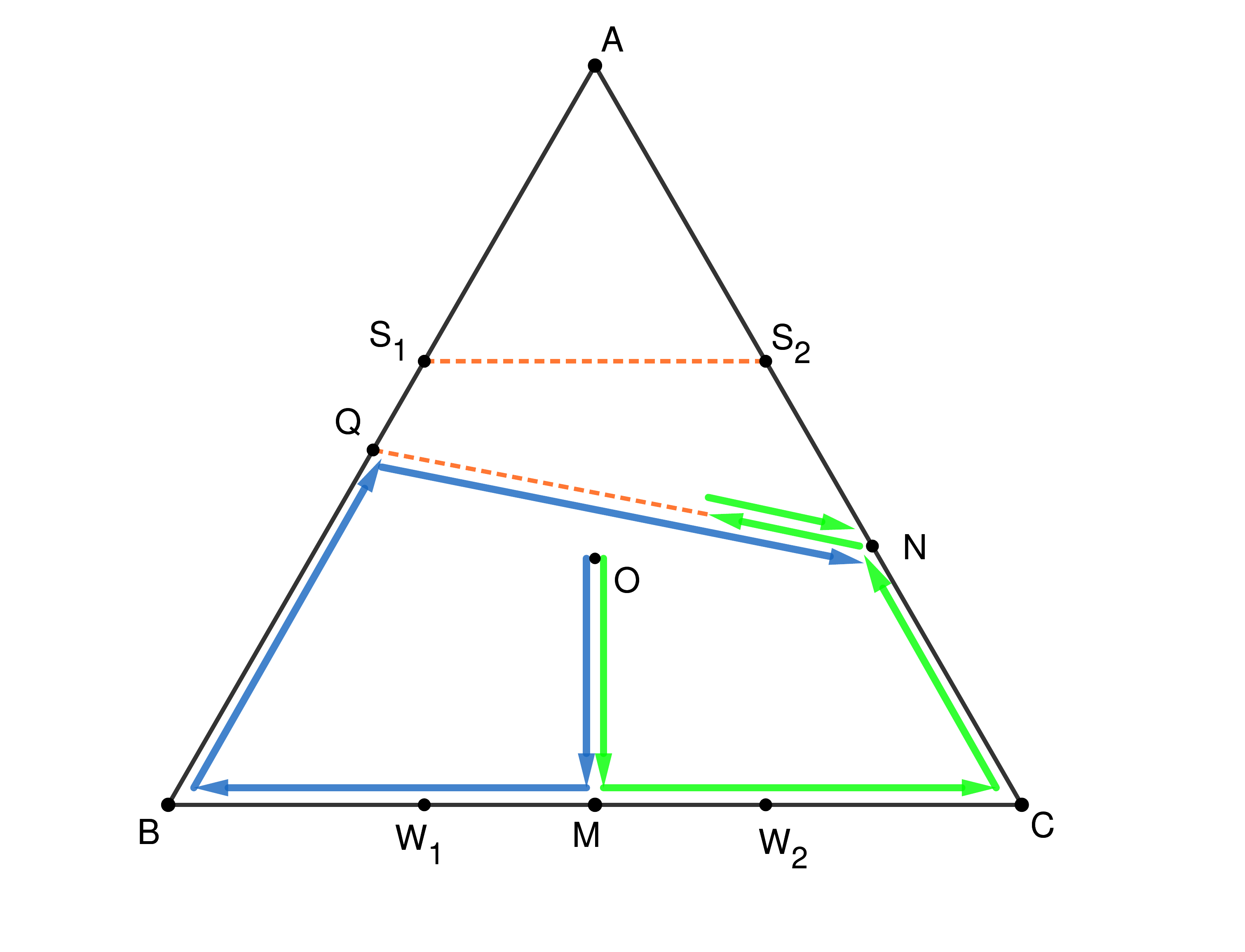}
  \caption{Exit is located on segment $CS_2$}
  \label{subfig:2r-nodetour-t3}
\end{subfigure}
\caption{Trajectories of agents based on the position of the exit point.}
\label{fig:4seg-traj}
\end{figure}

\begin{lemma} \label{lem:nodetour-segment1}
  Vertex $C$ is the critical point on segments $MC$.
\end{lemma}

\begin{proof}
  Suppose the exit is located at some point $N\neq C$ on segment $MC$,
   see Figure \ref{subfig:2r-nodetour-t2}.
  We only need
  to consider the case when the agents are at distance greater that $r$. 
  Since at that point the agents are moving in opposite directions on the
  segment $BC$, it is obvious that $R_1$ will be $r$-intercepted while travelling on edge $BA$. Hence according to Lemma \ref{lem:betagamma}, since $\beta> {\pi}/{2}$,
 point $N$ cannot yields the maximum evacuation time since there exists another point to the right of $N$ with a larger evacuation time. We conclude that C is the critical point in this segment. 
 \end{proof}

\begin{lemma} \label{lem:nodetour-segment3}
	On segment $CS_2$ see Figure \ref{subfig:2r-nodetour-t3}, vertex $C$ is the critical point.
\end{lemma}
\begin{proof}
	Let the exit be at some point $N$ on segment $CS_2$ and the $r$-interception point be $Q$ on side $BA$, see Figure \ref{subfig:2r-nodetour-t3}. 
		We know that $\beta +\gamma =\frac{2\pi}{3}$,
		and for all points on segment $CS_2$, the angle $\beta$ is between $0$ and $\nicefrac{\pi}{3}$.
		Then we get $2\cos(\beta)+\cos(\gamma)=2\cos(\beta)-\cos(\nicefrac{\pi}{3}+\beta)$ which is strictly greater than one.
		Hence by Lemma \ref{lem:betagamma}, placing the exit at a point closer to $C$ will result in higher evacuation time.
		We can conclude that vertex $C$ is the critical point in this segment.
\end{proof}


\begin{theorem} \label{th:nodetour-c}
  $E_{\mbox{\NoD}}=y+0.5+r+\frac{2(1-r^2)}{2r+1}$.
\end{theorem}
\begin{proof} We established that $C$ is the critical point for $MC$ and $CA$.
	When the exit is located at $C$, the evacuation time will be $t=y+0.5+|BQ|+|QC|$ where $Q$ is the point that $R_1$ is $r$-intercepted. Since both agents travel equal distances at the point of interception, we get $|BQ|=|QC|-r$. On the other hand by using the Cosine Rule we have $|QC|=\sqrt{BQ^2+1-BQ}$. 
	By solving for $|BQ|$
            and substituting in the expression for the evacuation time $t$,
        we obtain 
	 $t=y+0.5+r+\frac{2(1-r^2)}{2r+1}$. 
	Observe that  for $r\in [0,1]$, the function $r+\frac{2(1-r^2)}{2r+1}$ is  decreasing,  with maximum and minimum of $2$ and $1$ respectively.
	Therefore $y+0.5+r+\frac{2(1-r^2)}{2r+1} \geq y+0.5+1=y+1.5$.
        so based on Lemmas \ref{lem:nodetour-segment1} and \ref{lem:nodetour-segment3}, this evacuation time is also larger than maximum evacuation time of segments $MW_2$ and $S_2A$.
As already mentioned above, if the exit is located on segment $S_2A$
then the evacuation time is independent of exit location and is equal to $y+1.5$.
We conclude that placing the exit at point $C$ results in maximum evacuation time of algorithm  ${\mathcal{A}_1}$, which gives us the following theorem.
\end{proof}

\subsection{Trajectories with  Detours} \label{sec:2r-detour}
In Theorem \ref{th:nodetour-c} we showed that placing the exit at point $C$ causes the maximum evacuation time when using the \NoD algorithm.  In this section we propose different trajectories that
improve the evacuation time by decreasing the evacuation time at $C$. These
trajectories are a modified version of \textit{Equal-Travel with Detour Algorithms} in \cite{ChuangpishitMNO17} for the face-to-face communication model,
i.e., for $r=0$. The inclusion of a detour in the trajectories of each agent
consist of the agent
 stopping exploration at some point of the perimeter and moving inside the triangle to improve  the evacuation time when the exit is located in some segments around $C$ or $B$. When the agents realize that the exit was not found in these segments, they return to the same point on the boundary where they left off and resume the exploration of the perimeter.

\subsubsection{The \OneD Algorithm:} \label{ssec:detour}

The trajectories are symmetric and thus we define the detour for $R_1$ only.
We fix point $Q_1$ on the side $AB$, see Figure \ref{subfig:2r-onedetour}. The exact location of this point will be specified later. 
  \begin{figure}
  \centering
  \includegraphics[scale=0.3]{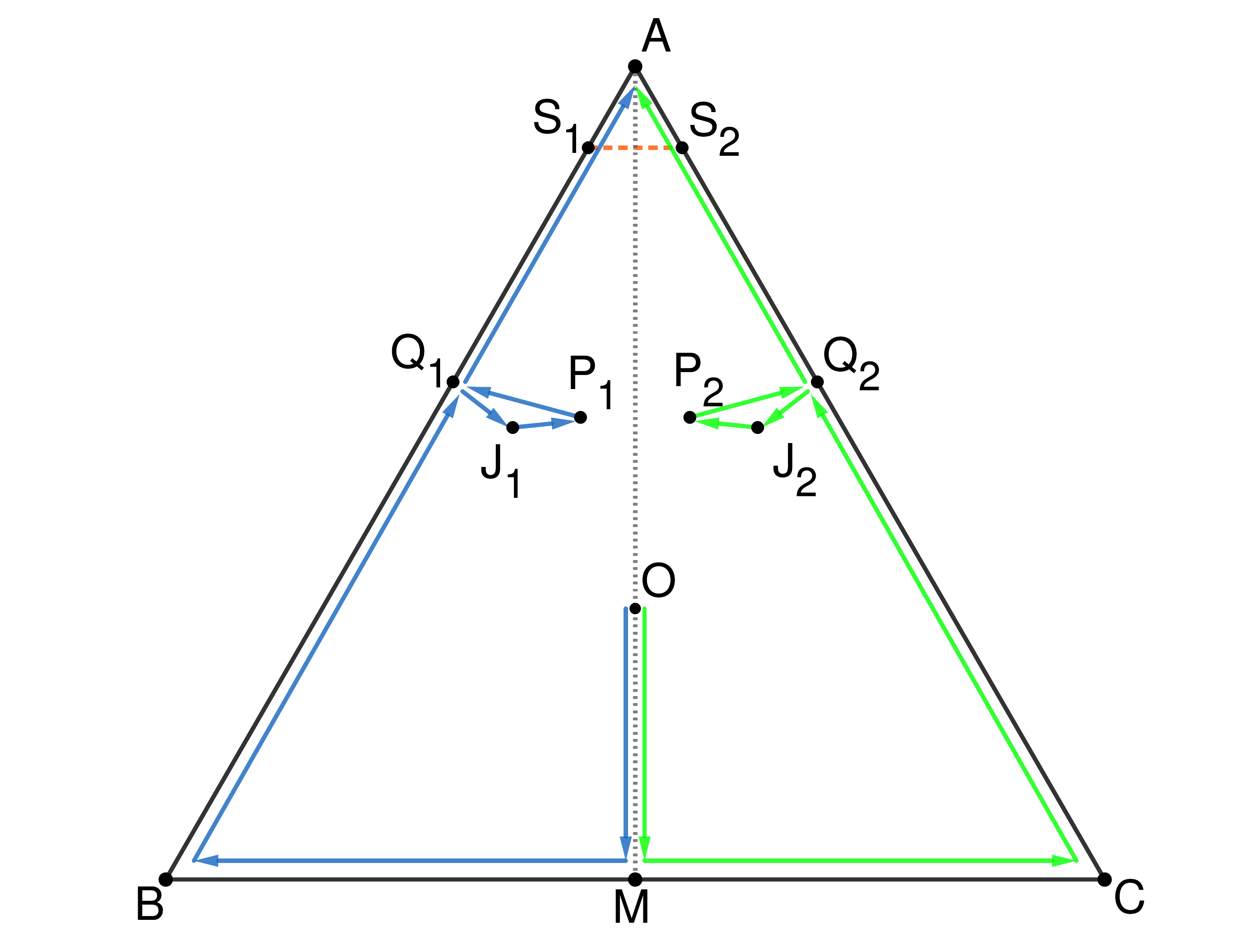}
  \caption{\OneD algorithm}
 \label{subfig:2r-onedetour}
  \end{figure}

Point $J_1$ is on segment $Q_1C$ such that it satisfies the equation
$|BQ_1|+|Q_1J_1|=|CJ_1|-r$.
Point $P_1$ is located on segment $J_1Q_2$ such that $P_1$ satisfies the equation 
$|Q_1J_1|+|J_1P_1|=|Q_2P_1|-r$.
Points $Q_2$, $J_2$ and $P_2$ are located symmetrically with those of points $Q_1$, $J_1$ and $P_1$ respectively, with respect to line $AM$.
\\
The trajectory of each agent is defined in Trajectories \ref{traj:2r1d}, see also Figure \ref{subfig:2r-onedetour}.

We show below that if $R_1$ reaches point $P_1$ and it is not notified about the exit by the other agent, then  it realizes that the exit has not been found yet. Thus it returns to point $Q_1$ where it started the detour
and resumes the exploration of  the perimeter. Algorithm \ref{alg:2r} with respect to Trajectory \ref{traj:2r1d} is referred to as the \OneD algorithm.
\begin{trajectory} \label{traj:2r1d}
$\ $\OneD \\
	\hspace*{3cm} $R_1:<O,B,Q_1,J_1,P_1,Q_1,A>$\\
	\hspace*{3cm} $R_2:<O,C,Q_2,J_2,P_2,Q_2,A>$
\end{trajectory}

\begin{lemma} \label{lem:onedetour-Q1J1}
  If the exit is located at some point $N$ on segment $MC$, then $R_1$ will be $r$-intercepted at or prior to reaching $J_1$.
\end{lemma}
\begin{figure}
  \vspace*{-8mm}
\begin{subfigure}{.5\textwidth}
  \centering
  \includegraphics[scale=0.22]{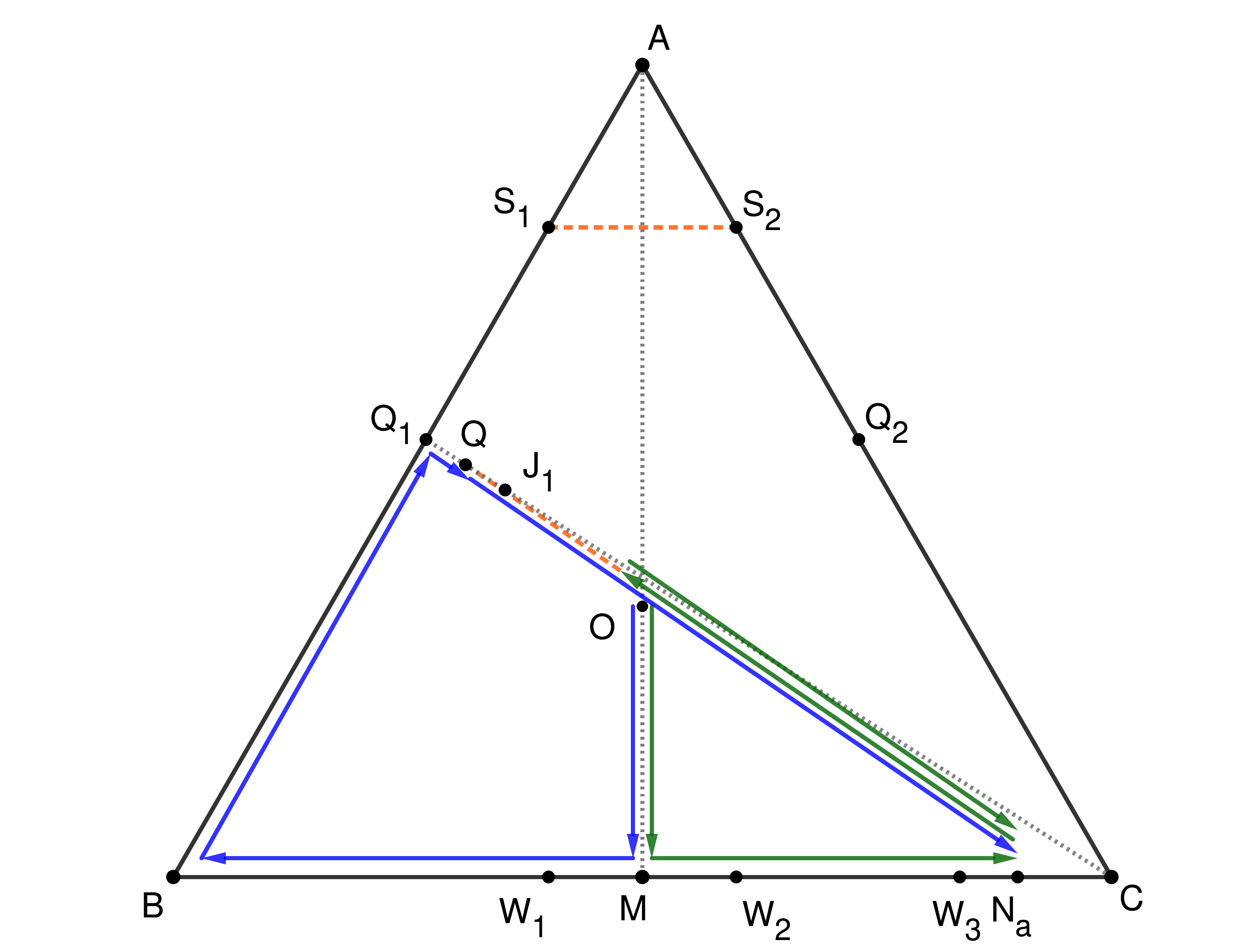}
  \caption{Exit is located on segment $MC$}
  \label{subfig:2r-onedetour-t3}
\end{subfigure}\\
\begin{subfigure}{.5\textwidth}
  \centering
  \includegraphics[scale=0.22]{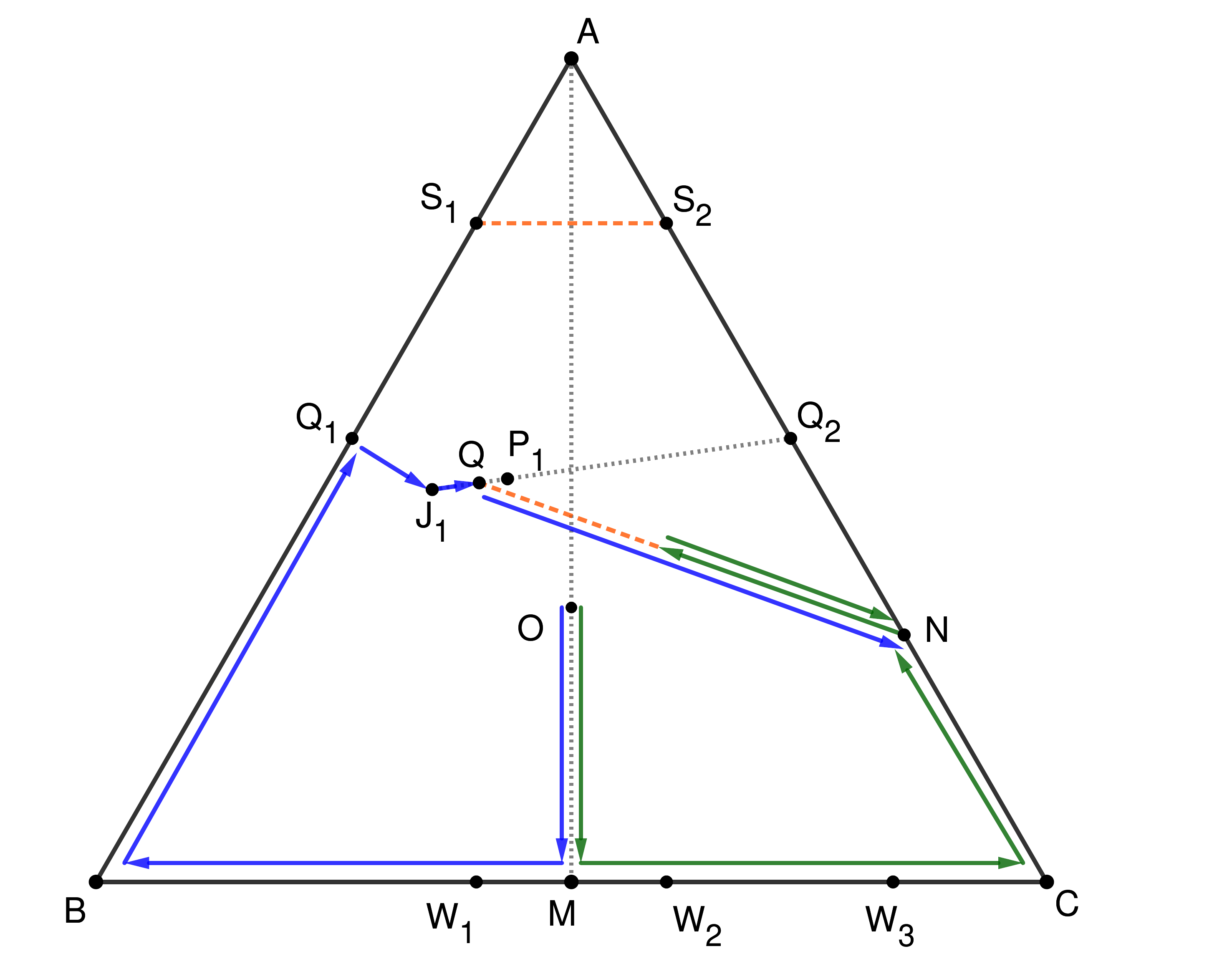}
  \caption{Exit is located on segment $CQ_2$}
  \label{subfig:2r-onedetour-t4}
\end{subfigure}
\begin{subfigure}{.5\textwidth}
  \centering
  \includegraphics[scale=0.22]{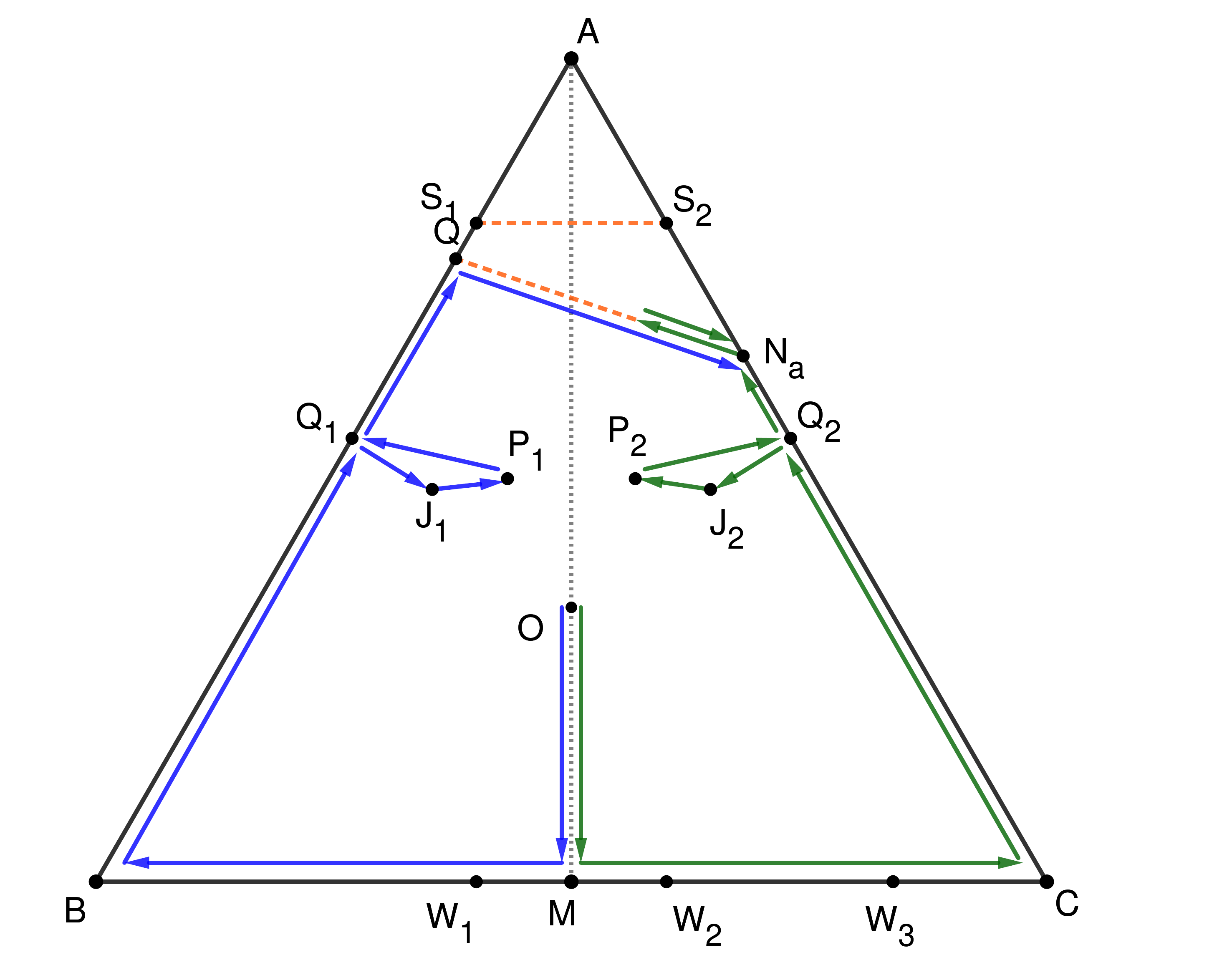}
  \caption{Exit is located inside segment $Q_2S_2$}
  \label{subfig:2r-onedetour-t5}
\end{subfigure}
        \caption{Trajectories in One -Detour algorithm, and trajectories of agents based on the position of the exit point in One Detour algorithm
        }
        \label{fig:2r-onedetour}
        \vspace*{-4mm}
\end{figure}
\begin{proof} (See Figure \ref{subfig:2r-onedetour-t3})
  If the exit is located at $C$ then, by the choice of $J_1$, agent  $R_1$ will be $r$-intercepted at point $J_1$. Hence if $N$ is located before $C$  then agent $R_2$ moves toward $J_1$ earlier and it  can $r$-intercept
 $R_1$  while it is on segment $Q_1J_1$.
\end{proof}

\begin{lemma} \label{lem:j1p1}
	Suppose the exit is located at some point $N$ on segment $CQ_2$, then $R_1$ will be $r$-intercepted while moving on segment $J_1P_1$.
\end{lemma}

\begin{proof} (See Figure \ref{subfig:2r-onedetour-t4})
	We know that if the exit is located at $Q_2$, then $R_1$ will be $r$-intercepted when it is at point $P_1$. In order to show that if the exit is before $Q_2$, agent $R_1$ can be intercepted before reaching $P_1$ it is enough to prove 
	$|CN|+|NP_1|-r \leq |CQ_2|+|Q_2P_1|-r$.
	For the purpose of contradiction suppose not, meaning 
	$|CN|+|NP_1|-r > |CQ_2|+|Q_2P_1|-r=|CN|+|NQ_2|+|Q_2P_1|-r$ and we get 
	$|NP_1|>|NQ_2|+|Q_2P_1|$ which according to the triangle inequality is impossible. Hence a contradiction.
\end{proof}

        We now split the trajectory of $R_2$ into segments $MC$, $CQ_2$, $Q_2S_2$, $S_2A$, 
        and determine the critical point for each segment and the evacuation time of the critical point for each segment.
\begin{lemma} \label{lem:onedetour-segment2}
On segment $MC$ point $C$ is the critical point, and the evacuation time for this segment is at most $y+0.5+|BQ_1|+|Q_1C|$.
\end{lemma}
\begin{proof} (See Figure \ref{subfig:2r-onedetour-t3}).
	For any arbitrary point on segment $MC$ (except the two endpoints), by Lemma \ref{lem:onedetour-Q1J1}, $R_1$ will be $r$-intercepted at some point $Q$ prior to  or at $J_1$, 
	and according to Lemma \ref{lem:betagamma}, since $\beta> {\pi}/{2}$, there exists another point in the direction of movement of $R_2$ such that placing the exit at that point will result in a higher evacuation time. We conclude that in segment $MC$, point $C$ is the critical point, and $R_1$ can reach $C$ in time  $y+0.5+|BQ_1|+|Q_1C|$.
\end{proof}

\begin{lemma} \label{lem:onedetour-segment4}
On segment $CQ_2$ point $C$ is the critical point.
\end{lemma}
\begin{proof} (see Figure \ref{subfig:2r-onedetour-t4}).
  Suppose the exit is located at some point $N$, then by Lemma \ref{lem:j1p1} $R_1$ will be $r$-intercepted while moving on segment $Q_1J_1$ at some point $Q$. Note that $\angle QQ_2N>\pi/6$, hence $\beta +\gamma <5\pi/6$. On the other hand we have $\angle QQ_2N+\beta+\gamma=\pi$. It is easy to see that $2\cos(\beta)+\cos(\gamma)$ is always greater than 1. By Lemma \ref{lem:betagamma}, there exists another point in the opposite direction of $R_2$ which yields a larger evacuation time if the exit is located there.
\end{proof}

\begin{lemma} \label{lem:onedetour-segment5}
  Assume the exit is located at point $N$ inside segment $Q_2S_2$ and let
  $Z$ be a point on segment $Q_1A$ such that $Q_1Z+r=Q_2Z$.
  Then the evacuation time for this exit is at most $y+0.5+|BQ_1|+|Q_1J_1|+|J_1P_1|+|P_1Q_1|+|Q_1Z|+|ZQ_2|$.
\end{lemma}
\begin{proof} (See Figure \ref{subfig:2r-onedetour-t5}).
Let the exit be at point $N$ on segment $Q_2S_2$ and the interception point be $Q$ on side $BA$. 
		We know that $\beta +\gamma =\frac{2\pi}{3}$,
		and for all points on segment $CS_2$, the angle $\beta$ is between 
		Then we get $2\cos(\beta)+\cos(\gamma)=2\cos(\beta)-\cos(\frac{\pi}{3}+\beta)$ which is strictly greater than one.
		Hence by Lemma \ref{lem:betagamma}, moving the exit point in the opposite direction of the movement of $R_2$ will result in higher evacuation time. Let $Z$ be the point on segment $Q_1A$ such that 
                $|Q_1Z|+r=|Q_2Z|$. Then we can conclude that the interception point for $N$ is in the segment $Q_1Z$ and therefore the evacuation time is at most the time needed to reach $Z$ plus $|ZQ_2|$,
which is equal to $y+0.5+|BQ_1|+|Q_1J_1|+|J_1P_1|+|P_1Q_1|+|Q_1Z|+|ZQ_2|$.
\end{proof}

As mentioned before, the evacuation time for an exit in segment $S_2A$ is
$t_{S_1}+r$, which is less than the evacuation time for exit located in segment $Q_2S_2$. 
Thus, combining the results of the previous lemmas  we can now give a value for $E_{\mbox{\OneD}}(2,r)$.

\begin{theorem}
	Let $t_1=y+0.5+|BQ_1|+|Q_1C|$ and  
	$t_2=y+0.5+|BQ_1|+|Q_1J_1|+|J_1P_1|+|P_1Q_1|+|Q_1Z|+|ZQ_2|$,
	where point $Z$ is the point that if the exit is located right after $Q_2$, agent $R_1$ will be $r$-intercepted at or before $Z$. Then $E_{\mbox{\OneD}}(2,r)=max\{t_1,t_2\}$.
\end{theorem}

Observe that by increasing the size of segment $BQ_1$, time $t_1$ increases, and on the other hand, decreasing length of $BQ_1$, increases $t_2$. Best value for $|BQ_1|$ is obtained when
$t_1=t_2$. Clearly, there is exactly one value of $Q_1$ which equates $t_1$ and $t_2$.
However, because of the complexity of the equations, we do not have an explicit solution for $Q_1$ as a function of $r$. 
We computed results for different values of $r$ which  are shown in Table 
\ref{tab:r2twodetour}.
As shown there, the \OneD algorithm with one detour has a lower evacuation time than the \NoD algorithm for $0 <r< 0.7$. 

\subsubsection{No detour for  $r\geq 0.7375$:} \label{ssec:large_r}
As can be seen from Table \ref{tab:r2twodetour}, the improvement
provided by using a detour in the evacuation algorithm diminishes when $r$ increases, and it does not give any improvement for $r=0.7375$.
We formalize this in the lemma below:

\begin{lemma}\label{lem:no_detour}
  No detour should be used for  $r>0.7374048$. 
\end{lemma}

\begin{proof}
  Values $r$ and $|Q_1J_1|$ are inversely related. Increasing $r$ will decrease the value of $|Q_1J_1|$ up to a point when $|Q_1J_1|$ is equal to zero. At this point we would have $r=\sqrt{|BQ_1|^2+1-|BQ_1|}-|BQ_1|$. By substituting this value in $f(r,|BQ_1|)=g(r,|BQ_1|)$ and solving that equation we get the values of $0.1843512042$ and $0.7374048168$ for $|BQ_1|$ and $r$ respectively. If we increase $r$, we get negative value for $|Q_1J_1|$ which is invalid.
\end{proof}
\subsubsection{The \TwoD Algorithm:}
It is shown in \cite{ChuangpishitMNO17}, that  for $r=0$, i.e., the face-to-face communication, the evacuation time can be improved by using more than one detour.  
We now show that for {\em smaller} values of $r$, a further improvement in evacuation time can be similarly achieved by making more detours. 
Consider the situation in the  execution of the \OneD algorithm
when $R_1$ and $R_2$  reach vertices $B$ and $C$ respectively, assuming no agent have found the exit so far. The remaining search problem will be a triangle with two unexplored sides of length 1, call this problem $\mathcal{P}_1$.

  \begin{figure}
 \centering
 	\includegraphics[scale=0.32]{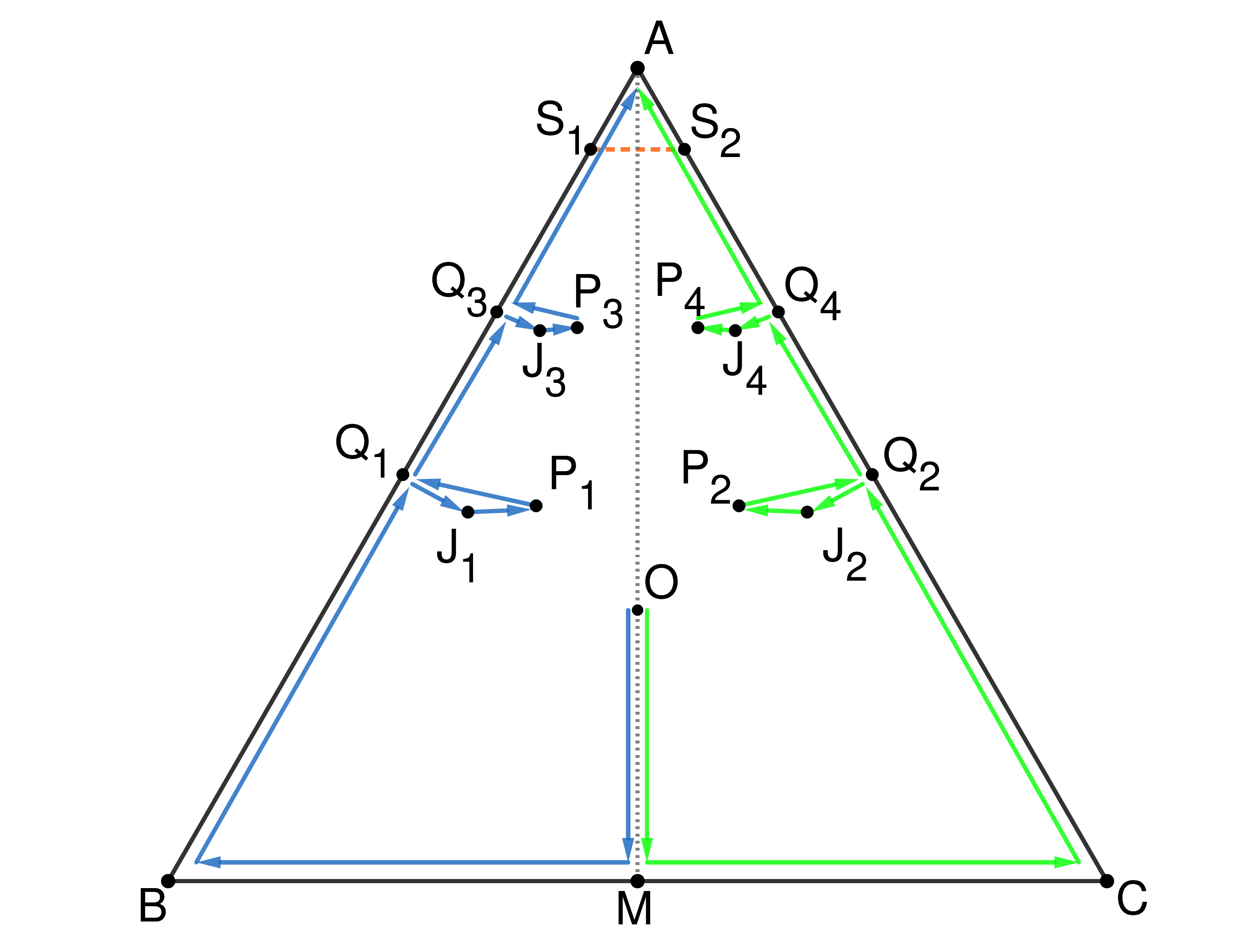}
	\caption{\TwoD algorithm}
 \label{subfig:2r-twodetours}
  \label{mixed}
\end{figure}

Now consider the time when the two agents finish their detour and get back to points $Q_1$ and $Q_2$ with no exit found.  Call the remaining search problem $\mathcal{P}_2$. It is obvious that $\mathcal{P}_2$ is a scaled down
version of $\mathcal{P}_1$, however with proportionally larger $r$.

Thus, if $r$ is not too large yet for problem $\mathcal{P}_2$, another detour could
 be done in the upper part of the triangle. The trajectory of two agents with two detours shown in Figure  \ref{subfig:2r-twodetours} is specified in Trajectories \ref{traj:2r2d}.
\begin{trajectory} \label{traj:2r2d}
$\ $ \TwoD \\
\hspace*{2cm}	$R_1:<O,M,B,Q_1,J_1,P_1,Q_1,Q_3,J_3,P_3,Q_3,A>$\\
\hspace*{2cm}	$R_2:<O,M,C,Q_2,J_2,P_2,Q_2,Q_4,J_4,P_4,Q_4,A>$
\end{trajectory}
Algorithm \ref{alg:2r} with respect to Trajectory \ref{traj:2r2d} is called \TwoD algorithm.
In the case of two detours, similarly as in the case of one detour,
it can be shown that there exists three critical points, namely $C$,
and the points right after $Q_2$ and $Q_4$. The evacuation times for these points will be as follows:
\begin{enumerate}
	\item $t_1=y+|MB|+|BQ_1|+|Q_1C|$
	\item$t_2=y+|MB|+|BQ_1|+|Q_1J_1|+|J_1P_1|+|P_1Q_1|
	+|Q_1Q_3|+|Q_3Q_2|$
	\item $t_3=y+|MB|+|BQ_1|+|Q_1J_1|+|J_1P_1|+|P_1Q_1|
	+|Q_1Q_3|+|Q_3J_3|+|J_3P_3|+|P_3Q_3|+|Q_3V|+|VQ_4|$\\
\end{enumerate}
\vspace{-6mm} \hspace{4.5mm} where $V$ is a point on segment $Q_3S_1$, such that $|Q_3V|=|VQ_4|-r$.\\

By equating these three values we obtain an optimized two detour evacuation algorithm.
It has been shown  in \cite{ChuangpishitMNO17} that for the face-to-face communication  detours can be recursively added to improve the evacuation time, though the improvement obtained by successive detours decreases rapidly.
In contrast, we showed above that for  $r> 0.7374$,  not even one detour improves the evacuation time. Similarly it can be shown that a  second detour is not helpful for $r>0.472504$.

\begin{remark}
  The size of the triangular detour(s) used above could be made slightly shorter by moving $P_i$ closer to $J_i$, which would improve the evacuation time slightly. 
  \end{remark}

\subsection{Lower Bound for Evacuating 2 Agents}
In this section we prove a lower bound for evacuation by two agents. The proof is essentially the same as the proof of the lower bound in \cite{ChuangpishitMNO17} for the case $r=0$, but needs to take into account the ability  of agents to communicate at distance $r$. We give the  proof here for completeness. First we need to generalize the Meeting Lemma used in \cite{ChuangpishitMNO17}. We say two points have opposite positions if one point is a vertex of $T$ and the other point is located on the opposite edge of that vertex.
\begin{figure} [h]
	\centering
	\includegraphics[scale=0.28]{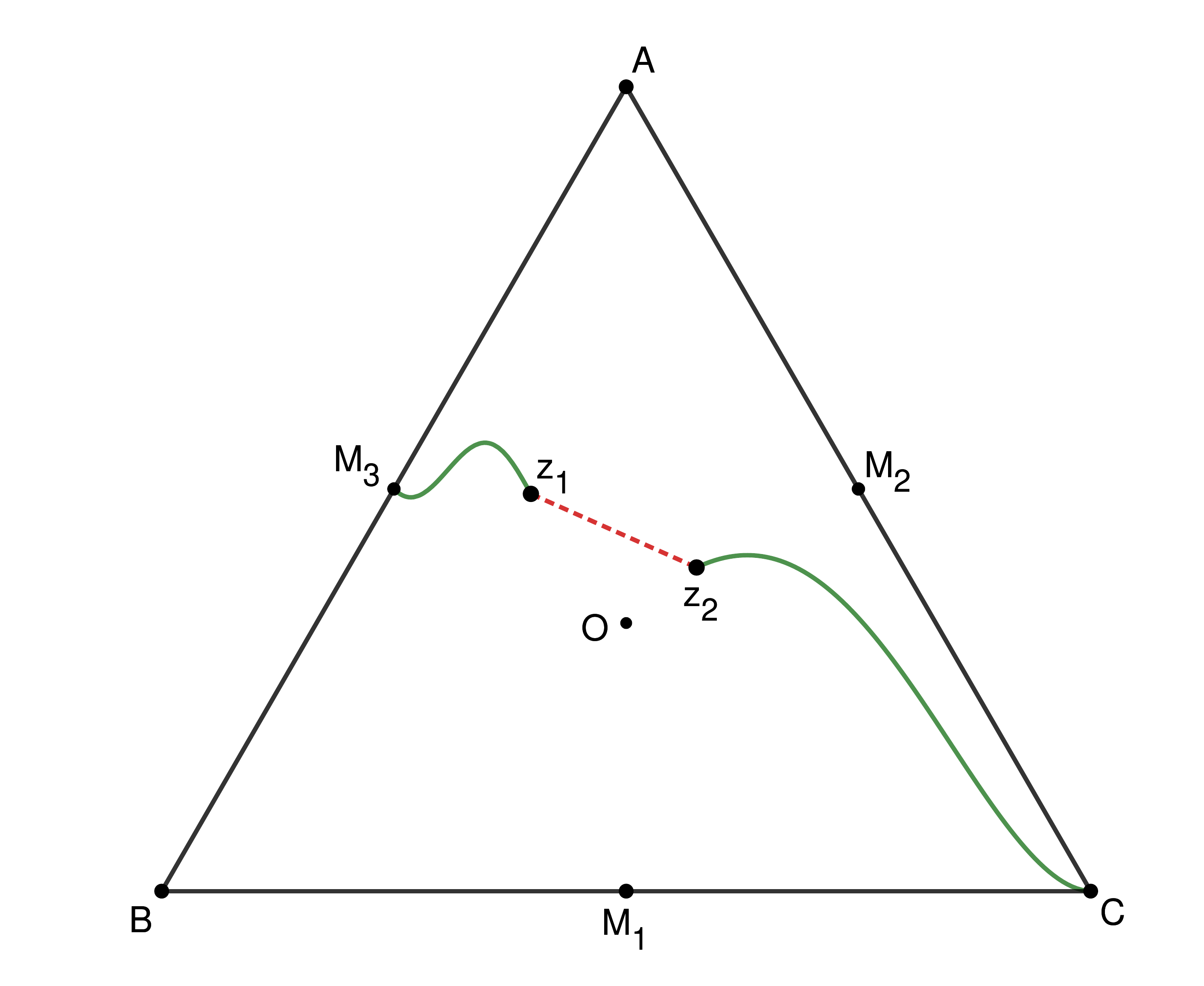}
	\caption{Modified Meeting Lemma}
	\label{fig:lower-bound-1}
\end{figure}
\begin{lemma} [Generalized Meeting Lemma]\label{lem:meeting_lemma}
Assume that $p_1,p_2 \in T$ and they have opposite positions, e.g. points $M_3$ and $C$ in Figure \ref{fig:lower-bound-1}. For any algorithm in which one of the agents visits $p_1$ at time $t'\geq 0.5+y$ and the other visits $p_2$ at time $t$ such that $t'<t<0.5+4y-r$, the two agents cannot communicate any information between time $t'$ and $t$.
\end{lemma}

\begin{proof}
Suppose the two agents exchange information between time $t'$ and $t$. They have to get close to each other, in order to communicate. 

Then there exists a time $t_z$ with $t ' \leq t_z \leq t$ such that $R_1$ is at point $z_1$ and $R_2$ is at $z_2$ at time $t_z$, and $|z_1  z_2| \leq r$.  Since $p_1$ and $p_2$ have opposite positions, $|p_1p_2|\geq 3y$. Therefore $|p_1z_1|+|z_1z_2|+|z_2p_2| \geq 3y$. On the other hand, we know that $t_z-t'\geq |p_1z_1|$ and $t-t_z\geq |p_2z_2|$. Combining these facts together, we obtain:
$$h\leq |p_1z_1|+|z_1z_2|+|z_2p_2| \leq t_z-t'+|z_1z_2|+t-t_z =t-t'+r$$
$$t-t'+r<0.5+h+y-r-t'+r \leq 0.5+h+y-0.5-y=h$$ 
which is a contradiction.
\end{proof}

Using a similar arguments used in \cite{ChuangpishitMNO17} we now establish the lower bound for the case $r\geq 0$ for two agents.

\begin{theorem}
  \label{th:lb2}
	Two agents are at a centroid of an equilateral triangle with sides 1. the evacuation time for two agents of transmission range $r$ positioned at the centroid of a triangle with sides one is at least $max\{1.5+y,1+4y-r\}$. 
\end{theorem}

\begin{proof}
	For the purpose of contradiction assume there exists algorithm $\mathcal{A}$ such that $E_{\mathcal{A}}(2,r)<1+4y-r$. Let us focus on the set of points $S=\{A,B,C,M_1,M_2,M_3\}$. We give an Adversary Argument. There exists some input $I$ in which the exit is the last of the point in $S$ visited by an agent. Suppose time $t$ is the time that the fifth point from the set $S$ is visited and $v_1$ through $v_6$ be the order that points are visited. Wlog we assume that $v_5$ is visited by $R_1$. Since at time $t$, the fifth vertex is visited, then 3 points must have been visited by one of the agents and $t\geq y+1$. On the other hand since the algorithm should satisfy $E_{\mathcal{A}}(2,r)<1+4y-r$, then $t<0.5+4y-r \leq 0.5+4y$ since $R_1$ needs extra time $0.5$ to get to the sixth vertex.
	\\
	We now examine the following exhaustive cases based on whether $v_5$ is a midpoint or a vertex. 
	\begin{description}
	\item[Case 1. Point $v_5$ is a vertex of $T$:] Wlog assume that $v_5$ is $C$. If $v_6$ is one of $A$, $M_3$ or $B$, then it takes at least $h$ for $R_1$ to evacuate the triangle and $E^*(2,r)\geq t+h\geq 1+4y$ which is a contradiction. We conclude that the $v_6$ should be either $M_1$ or $M_2$. Note that $R_1$ can have visited at most one of $A$, $M_3$ and $B$ by time $t$, hence $R_2$ should have visited at least two points of $A$, $M_3$ and $B$. Assume $v$ is the second vertex of the set of $A$, $M_3$ and $B$ visited by $R_2$ at time $t'$. Clearly $t'\geq 0.5+y$. By the Generalized Meeting Lemma the two agents cannot communicate between $t'$ and $t$ on input $I$.
	\\
	Now consider input $I'$ in which the exit is located at $v$. On this input $R_1$ and$R_2$ behave identical to input $I$ until time $t'$. After this time $R_2$ may try to $r$-intercept $R_1$ but by Modified Meeting Lemma we know that the $r$-interception does not occur before time $t$. Hence $R_1$ has to travel at least $h$ to get to the exit which indicates that evacuation time will be at least $1+y+h=1+4y$ on input $I'$, a contradiction. 

	\item[Case 2. $v_5$ is a midpoint of a side of $T$, and $v_6$ is another midpoint:]
	Wlog we assume that $v_5$ is $M_2$ and $v_6$ is $M_3$. If $R_1$ visits two vertices before arriving at $M_2$ then $E^*(2,r)\geq 2y+1.5$ which is a contradiction. We conclude that $R_2$ must have visited 2 vertices before $t$. It is obvious that $R_2$ cannot visit the second of these two vertices sooner than time $2y+1$. If the second vertex is $B$ then the adversary places the exit at $M_2$ and if the second vertex is $C$, it will place the exit at $M_3$. In both cases $E^*(2,r)\geq 2y+1+3y>1+4y$ which is a contradiction. 
	We conclude that the second vertex visited by $R_2$ must be $A$.
	\\
	 Observe that if $M_1$ is visited by $R_2$, then $E^*(2,r)\geq y+2 >1+4y$. Therefore $R_1$ should visit $M_1$ as well as either $B$ or $C$ before arriving at $M_2$. Let $P$ be the second point from set $S$ visited by $R_1$ at time $t_1$. Clearly $t_1\geq y+0.5$. On the other hand, $R_2$ must visit $A$ before time $0.5+4y-r$. Let this time be $t_2$. Clearly $t_1<t_2$, since if not, the time $R_1$ gets to $M_2$ will be at least $t_2+0.5\geq 2y+1.5$ which is a contradiction. By the modified meeting lemma, $R_1$ and $R_2$ cannot exchange information between $t_1$ and $t_2$. Now consider input $I'$ in which the exit is located at $P$. Agent $R_2$ has the same behaviour until it reaches point $A$ at time $t_2\geq 2y+1$ and has to travel at least $3y$ to get to the exit. Hence $E^*(2,r)\geq 2y+1+3y>1+4y$. A contradiction.

	\item[Case 3. $v_5$ is a midpoint of a side of $T$, and $v_6$ is a vertex:]
	Wlog assume $v_5$ is $M_2$, if $v_6$ is $B$ then $E^*(2,r)\geq 4y+1$ which is a contradiction. We conclude that $v_6$ is either $A$ or $C$. Wlog we assume $v_6$ is $A$. If a single agent visits both $B$ and $C$, it takes time at least $2y+2$ to get to $A$, hence $B$ and $C$ should be visited by different agents. Now we consider 2 cases: $R_1$ visits $C$ and $R_2$ visits $B$; and $R_1$ visits $B$ and $R_2$ visits $C$.
	\\
	Suppose $R_1$ visits $C$. First observe that $R_1$ cannot also visit $M_3$ before reaching point $M_2$ as doing so will result in $E^*(2,r)\geq 4y+1$ which is a contradiction. Therefore $M_3$ should be visited by $R_2$. If $M_1$ is visited by $R_2$, Lemma \ref{lem:lem5} assures that $E^*(2,r)\geq 4y+1$ and if $M_1$ is visited by $R_1$, then Lemma \ref{lem:lem6} assures the same thing.
	\\
	Now suppose $R_1$ visits $B$ before reaching point $M_2$ and $R_2$ visits $C$. It is obvious that $R_1$ cannot visit both $M_3$ and $M_1$ as by doing so, it will take at least $y+1.5>4y+0.5$ to reach $M_2$. 
	Lemma \ref{lem:lem7-9} now assures that $E^*(2,r)\geq 4y+1-r$ .
	\end{description}
\end{proof}
	For Lemmas \ref{lem:lem5} through \ref{lem:lem7-9}, we assume that $v_5=M_2$ and $v_6=A$, and $R_1$ visits $M_2$ at time $1+y\leq t_{M_2}<4y+0.5-r$. We use the following observation from \cite{ChuangpishitMNO17}.
	
\begin{obs} \label{obs:Ap}
	\cite{ChuangpishitMNO17} Let $p$ be a point on the boundary. If at time $1+4y-|Ap|$, both $A$ and $p$ are unvisited, then $E^*(2,r)\geq 1+4y$
\end{obs}

\begin{lemma} \label{lem:lem5}
	If $R_2$ visits $B$, $M_3$ and $M_1$, and $R_1$ visits $C$ and $M_2$, then $E_{\mathcal{A}}(2,r)\geq 1+4y-r$.
\end{lemma}

\begin{figure}[ht]
\centering
\begin{subfigure}{.5\textwidth}
  \centering
  \includegraphics[scale=0.26]{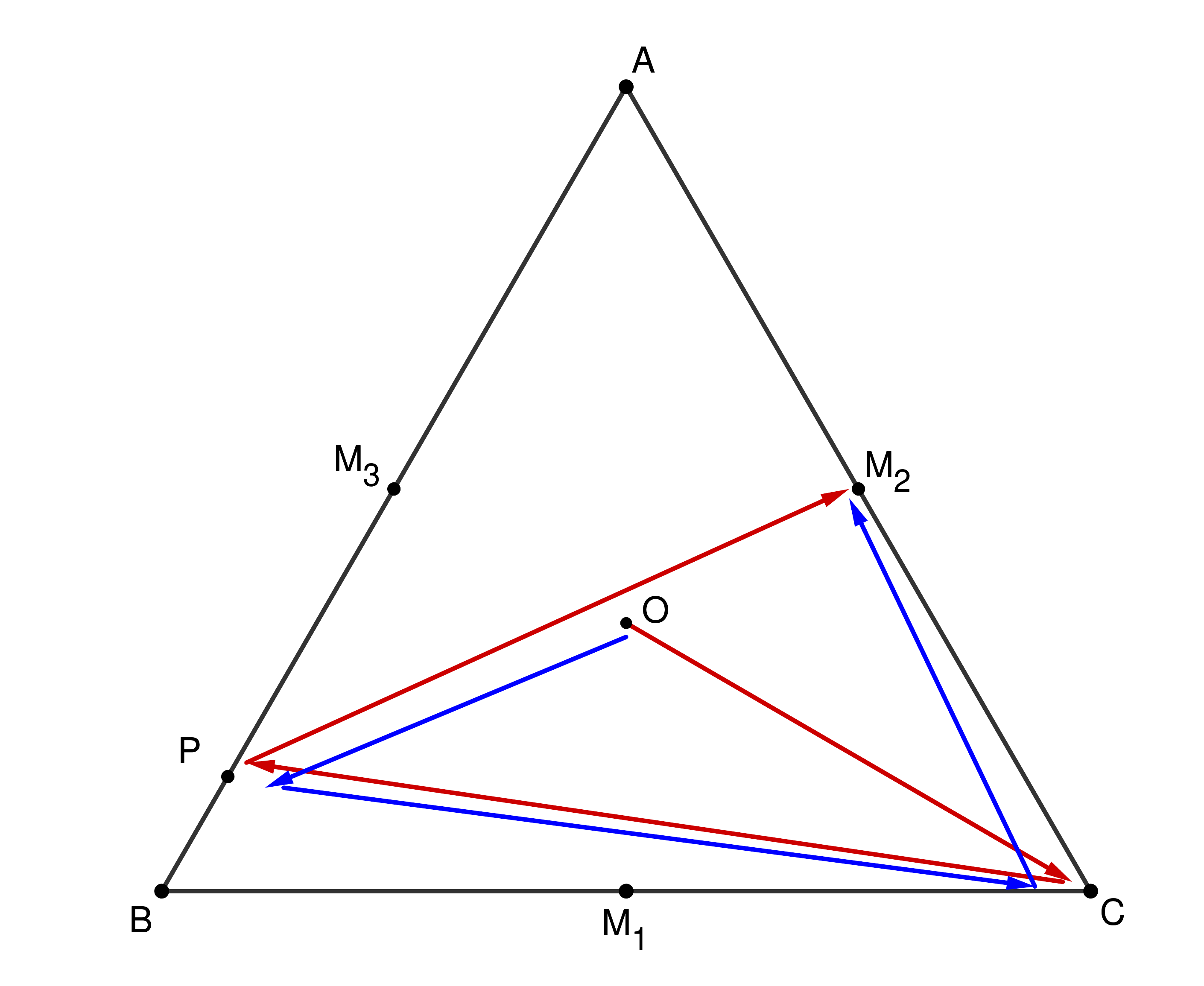}
  \caption{}
  \label{subfig:lemma-5-a}
\end{subfigure}%
\begin{subfigure}{.5\textwidth}
  \centering
  \includegraphics[scale=0.26]{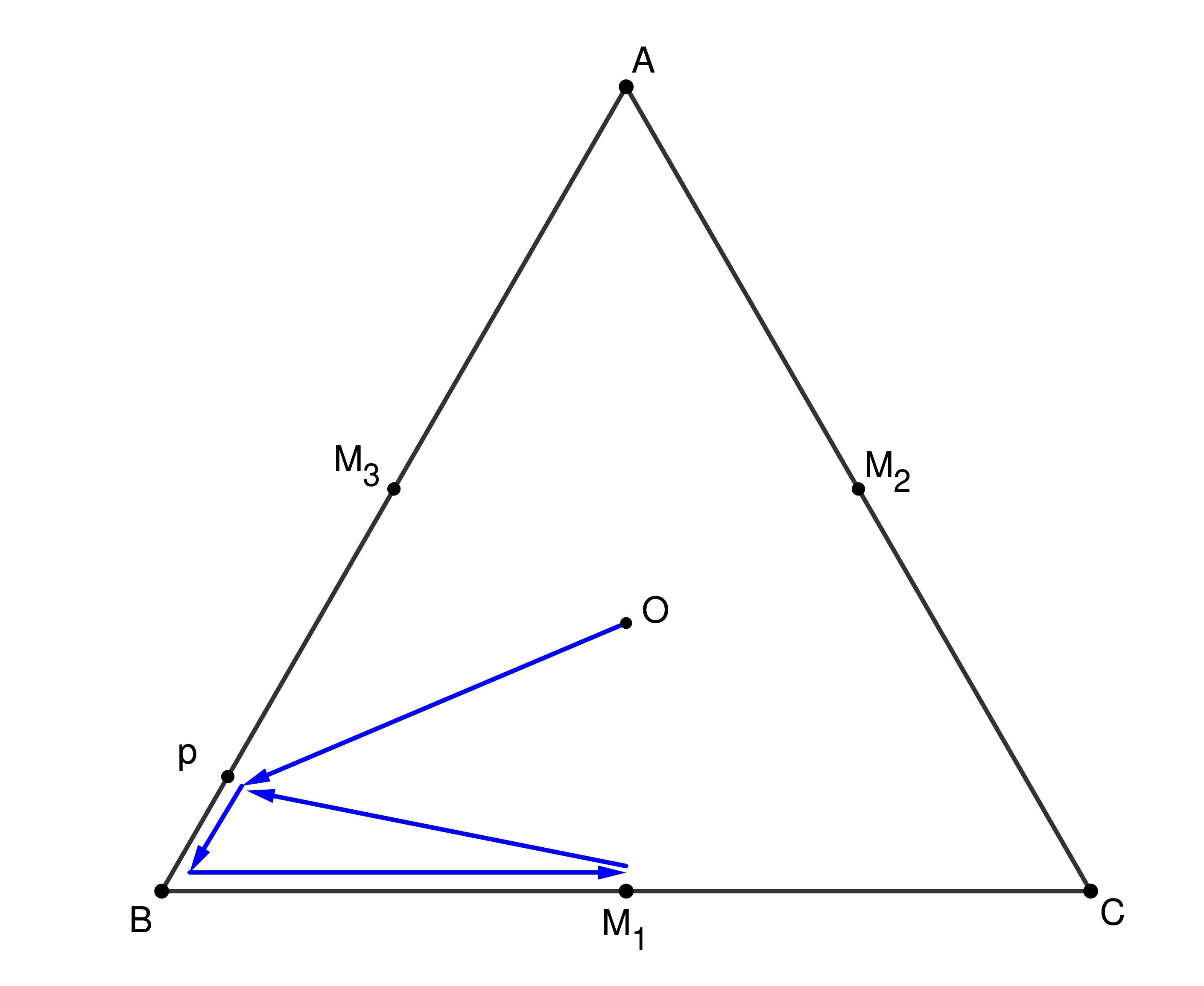}
  \caption{}
  \label{subfig:lemma-5-b}
\end{subfigure}
\caption{Illustration of (a) possible trajectories of $R_1$ and (b) trajectory of $R_2$ in support of Lemma \ref{lem:lem5} }
\label{fig:lemma-5}
\end{figure}
		
\begin{proof}  Identical to the proof of Lemma 5 in
  \cite{ChuangpishitMNO17}.
\end{proof}
\begin{lemma} \label{lem:lem6}
	If $R_2$ visits $B$, $M_3$, and $R_1$ visits $C$, $M_1$ and $M_2$, then $E_{\mathcal{A}}(2,r)\geq 1+4y-r$.
\end{lemma}

\begin{figure} [h]
	\centering
	\includegraphics[scale=0.28]{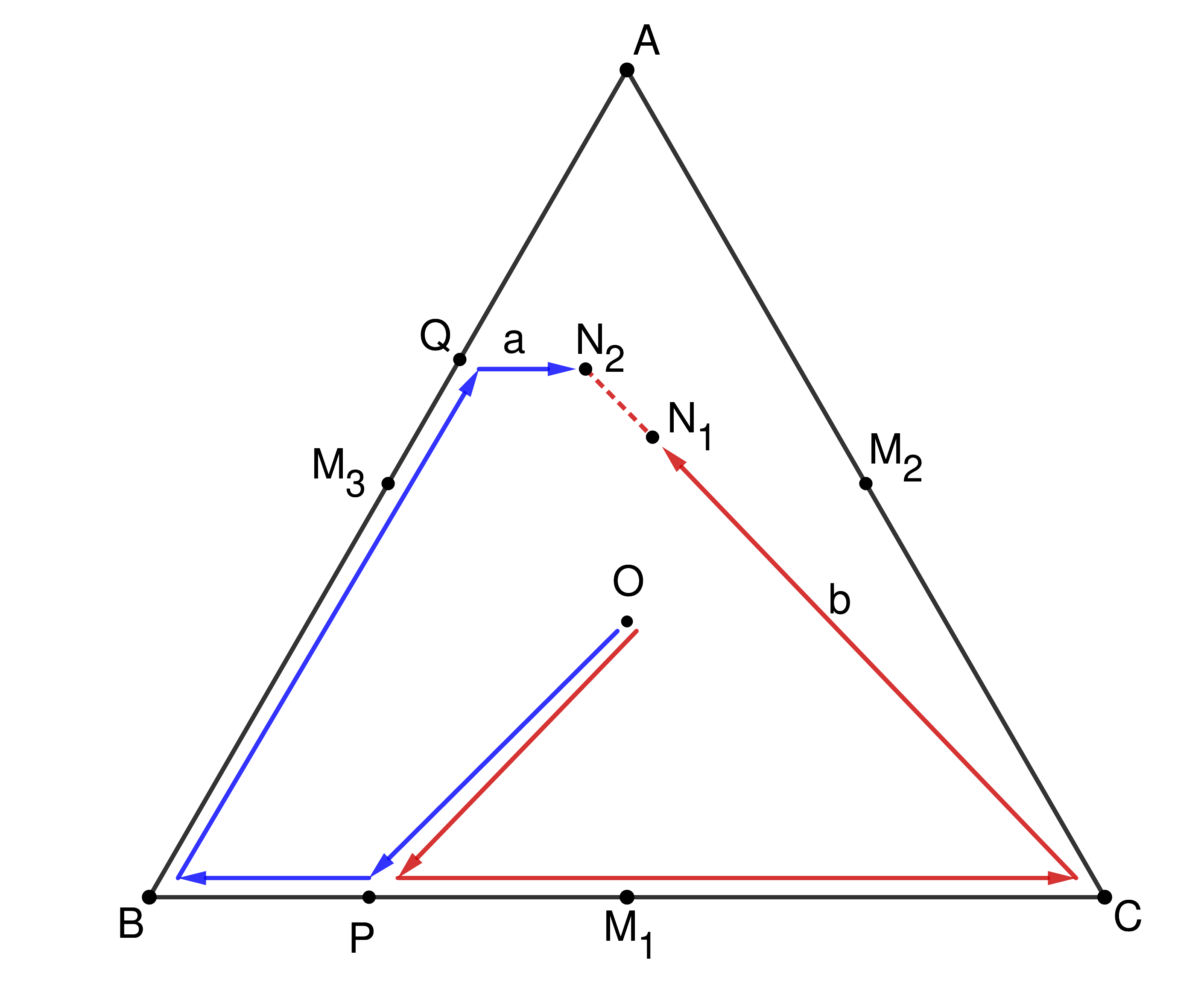}
	\caption{Trajectories when $R_2$ visits $B$ before $M_3$ and $R_1$ visits $M_1$ before $C$}
	\label{fig:lemma-6}
\end{figure}

\begin{proof}
	Recall that $t_{M_2}\geq 1+y$. If $R_1$ and $R_2$ don't meet between $t_B$ and $t_{M_2}$ then by placing the exit at $B$, it will take at least $1+4y$ for $R_1$ to get to the exit. So we conclude that they should communicate with each other between $t_B$ and $t_{M_2}$.
	If $R_2$ visits $M_3$ before $B$, then $t_B\geq 0.5+y$ and since $t_{M_2}<0.5+4y-r$, by the Generalized Meeting Lemma they could not exchange any information between $t_B$ and $t_{M_2}$. Therefore $R_2$ must visit $B$ before $M_3$.
	
	Now if $R_1$ visits $C$ before $M_1$, using a similar argument as in Lemma \ref{lem:lem5}, it can be verified that an unvisited point $P$ exists on segment $CM_2$ at time $1+4y-|AP|$. It follows from the Observation \ref{obs:Ap} that $E_{\mathcal{A}}(2,r)\geq 1+4y$. Hence $M_1$ should be visited before $C$.
	
	Observe that $t_C\geq 0.5+y$ and $t_{M_3} <0.5+4y-r$, so by the Generalized Meeting Lemma $R_1$ and $R_2$ cannot exchange information between time $t_C$ and $t_{M_3}$. We conclude that they should exchange information after $t_C$ and after $t_{M_3}$. At this point, the whole segment $BM_1$ must be visited; because if not, before time $t_{M_3}+3y/2$ there exist an unvisited point on $BM_1$ segment with distance at least $3y$ from $A$ and by Observation \ref{obs:Ap}, $E_{\mathcal{A}}(2,r)\geq t_{M_3}+9y/2\geq 0.5+13y/2$. Thus there must be a point $P$ on segment $BM_1$ such that $BP$ is explored by $R_2$ and $CP$ is explored by $R_1$. Thus $t_B\geq |OP|+|BP|$ and $t_C\geq |OP|+(1-|BP|)$. Suppose the exit is located at $C$, if the agents don't communicate before $R_2$ reaches $A$, then clearly $E_{\mathcal{A}}(2,r)\geq 1+4y$. So they have to exchange information before $R_2$ reaches $A$. Let $Q$ be the last point visited by $R_2$ on segment $AB$, before either it gets $r$-intercepted by $R_1$ or moves inside $T$ in order to get closer to the other agent. Let the interception points be $N_1$ and $N_2$, see Figure \ref{fig:lemma-6}, and let $a=|QN_2|$ and $b=|N_1C|$. Clearly points $Q$ and $N_2$ will merge and $a=0$, if $R_2$ is $r$-intercepted on segment $AB$.
	
	Note that $R_1$ reaches point $N_1$ at the same time that $R_2$ reaches point $N_2$ and they are $r$-apart. Observe that:
	
	\begin{description}
		\item[] For $R_2$ to get to the exit at $C$ in time, it must be that $t_{N_2}+r+b<1+4y-r$ and since $t_{N_1}\geq t_C+b$, we obtain:
		\begin{equation} \label{eq:b}
			b<\frac{1+4y-2r-t_C}{2}\leq \frac{1+4y-2r-|OP|-1+|BP|}{2}
		\end{equation}
		
		\item[] Now let $Q'$ be a point infinitesimally close to $Q$ on segment $QA$ and visited at time $t_{Q'}$. Then by Observation \ref{obs:Ap}, we have $t_{Q'}<1+4y-r-|AQ'|$. Also we have $t_{Q'}\geq t_{N_2}+a\geq t_B+|BQ'|+2a$, which implies that 
		\begin{equation} \label{eq:a}
			a<\frac{1+4y-r-|AQ'|-|BQ'|-t_B}{2}\leq \frac{1+4y-r-1-(|OP|+|BP|)}{2}
		\end{equation}
	\end{description}
		Note that $a+b+r\geq 3y$, yet from inequalities \ref{eq:b} and \ref{eq:a} we get:
		$$a+b+r<4y-|OP|-r/2\leq 3y$$
		which is a contradiction. We conclude that these $r$-interception points do not exist and if the exit is located at $C$, agent $R_2$ cannot get to the exit point in time.

\end{proof}
\begin{lemma} \label{lem:lem7-9}
	If $R_2$ visits $C$ and at least one of points $M1$ or $M_3$, and $R_1$ visits $B$ and at most one of points $M_1$ or $M_3$, then $E_{\mathcal{A}}(2,r)\geq 1+4y-r$.
\end{lemma}

\begin{proof}
  Identical to the proof of Lemma 7 in
  \cite{ChuangpishitMNO17}.
\end{proof}

The following graph and Table \ref{tab:r2twodetour} illustrates the evacuation time of two agents for different values for $r$ and different algorithms.

\begin{figure}[h]
  \centering
  \includegraphics[scale=0.35]{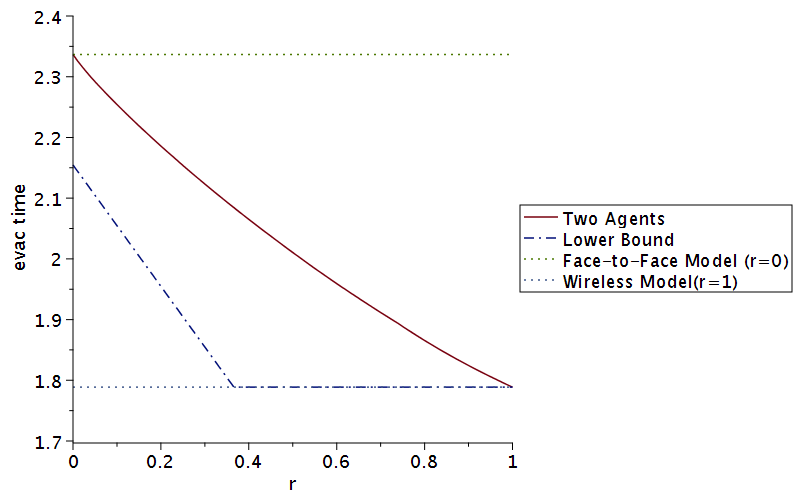}
 \caption{Comparison of evacuation time of 2 agents with transmission range of $r=0$, $r<1$ and $r=1$}
 \label{fig:2r-comparison}
\end{figure}

\begin{table}

	\begin{center}
		\begin{tabular}{ |c|c|c|c|c| } 
		\hline
		$r$ & Ev. Time & 	Ev. Time &  Ev. Time  & Lower bound\\                 &  of \NoD& of \OneD& of \TwoD &  Th. \ref{th:lb2}\\
\hline
	0.10 & 2.25424 & 2.27422 & 2.53867& 2.0547 \\
	0.20 & 2.18584& 2.19427 &  2.36010&2.0447\\
	0.30 & 2.12325& 2.12651 & 2.22617& 2.0347\\
	0.40 & 2.06506& 2.06593 & 2.12200&N/A \\
        0.50 & N/A & 2.01050&2.03867& N/A\\
        0.60 & N/A & 1.95926&1.97049&N/A\\
        0.70 & N/A & 1.91169&1.91367&N/A\\
	\hline
	\end{tabular}
\end{center}
	\caption{Evacuation times of 2 Agents using algorithms with  two, one or no detours.}
	\label{tab:r2twodetour}
\end{table}

\section{Evacuation of Three or Four Agents}\label{sec:3a}

Algorithms for the evacuation of three agents from the centroid have been previously proposed for both 
 $r=0$ in \cite{ChuangpishitMNO17}, and $r=1$ in
\cite{EtrSq}. These two algorithms use very different trajectories. Our \XthreeC and \XoneC strategies described in Section \ref{sec:4-1} and \ref{sec:4-2} can be considered generalizations of the algorithm for 3 agents for the $r=0$ case in  \cite{ChuangpishitMNO17} and the $r=1$ case in \cite{EtrSq} respectively. However, the best partitioning of the perimeter into segments is non-trivial to find, and requires significant experimentation. Additionally, the positions where the agents should connect after their initial exploration is also not obvious for arbitrary $r$, while it should clearly be the centroid for the $r=0$ case and no meeting is required for the $r=1$ case.

We describe the \XthreeC algorithm for 3 agents for the case $r>0$ 
in Section \ref{sec:4-1} 
and  the \XoneC algorithm for three agents in Section \ref{sec:4-2}. Finally in Section~\ref{sec:4-agents}, we briefly describe the algorithm for 4 agents. 
An important consideration in the design of algorithms for three agents 
is the fact that we can use one of them as a {\em relay} which can extend the range at which an agent can send a message with the location of the exit.
Both algorithms that are proposed in the following subsections follow the generic Algorithm \ref{alg:3r}, but they differ in the trajectories assigned to each agent.

\begin{algorithm}[h]
	\caption{Evacuation Algorithm for Three and Four Agents.}\label{alg:3r}
		\begin{algorithmic}
			\Function{Exploration}{}
				\State found$\gets$ false
				\While{not\textless found\textgreater \, and not\textless msg\_recd\textgreater}
					\State move along the predetermined trajectory
				\EndWhile				
				\State \Call{action}{}
			\EndFunction
		\end{algorithmic}
		\begin{algorithmic}
			\Function{Action}{}
				\If{found}
					\State P $\gets$ current location
					\While{the other two agents are not in effective communication range}
						\State continue moving on the trajectory
					\EndWhile
				\EndIf
				\State broadcast $<P>$
				\State go to $P$ and exit
			\EndFunction
		\end{algorithmic}
\end{algorithm}

\subsection{\em Explore 3 sides before Connecting (X3C)}  \label{sec:4-1}
\begin{figure}[h]
\begin{subfigure}{0.5\textwidth}
	\centering
	\includegraphics[scale=0.23]{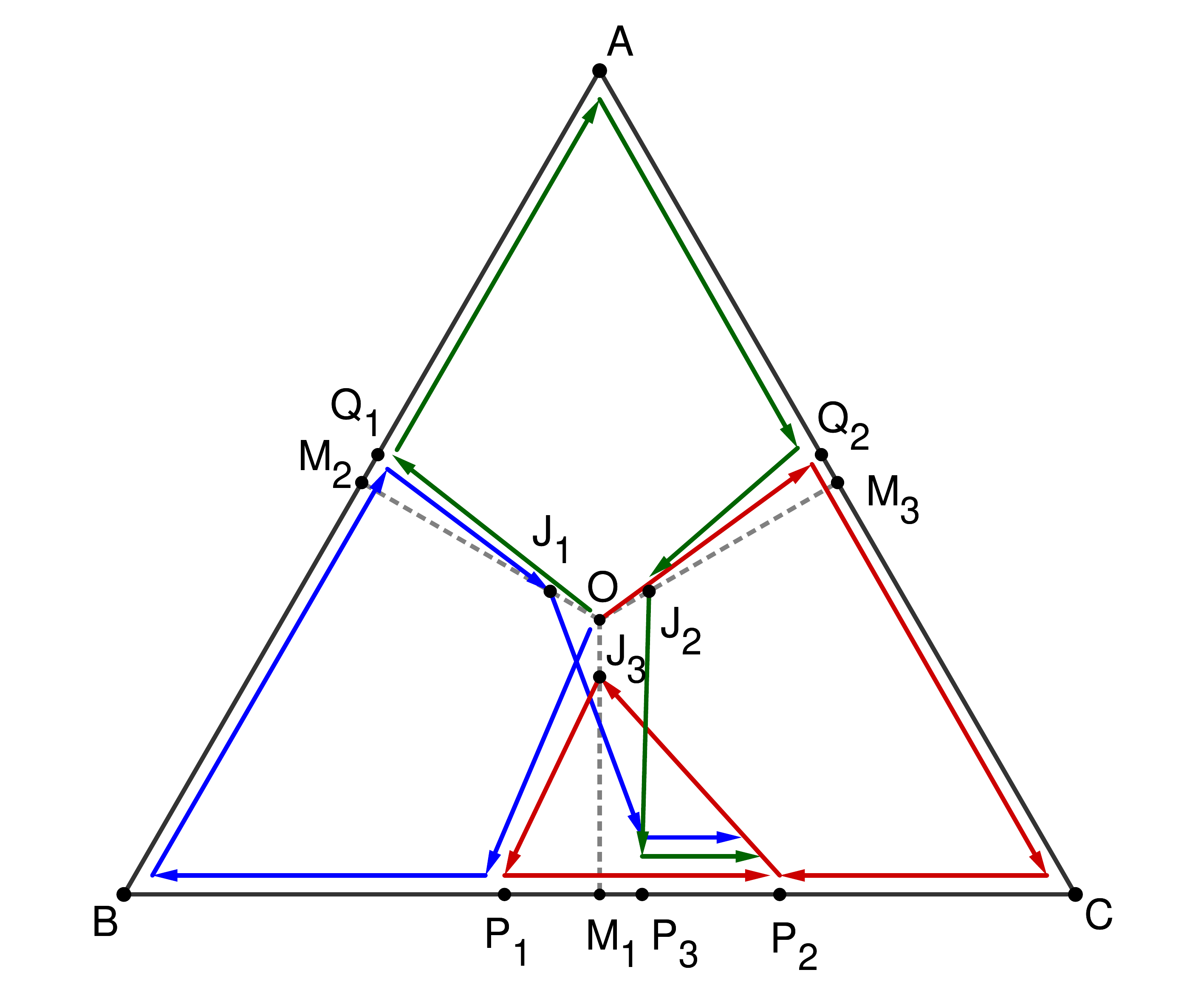}
	\caption{ \XthreeC Trajectories for 3 agents}
	\label{fig:r3asym}
\end{subfigure}
\begin{subfigure}{0.5\textwidth}
	\centering
	\includegraphics[scale=0.23]{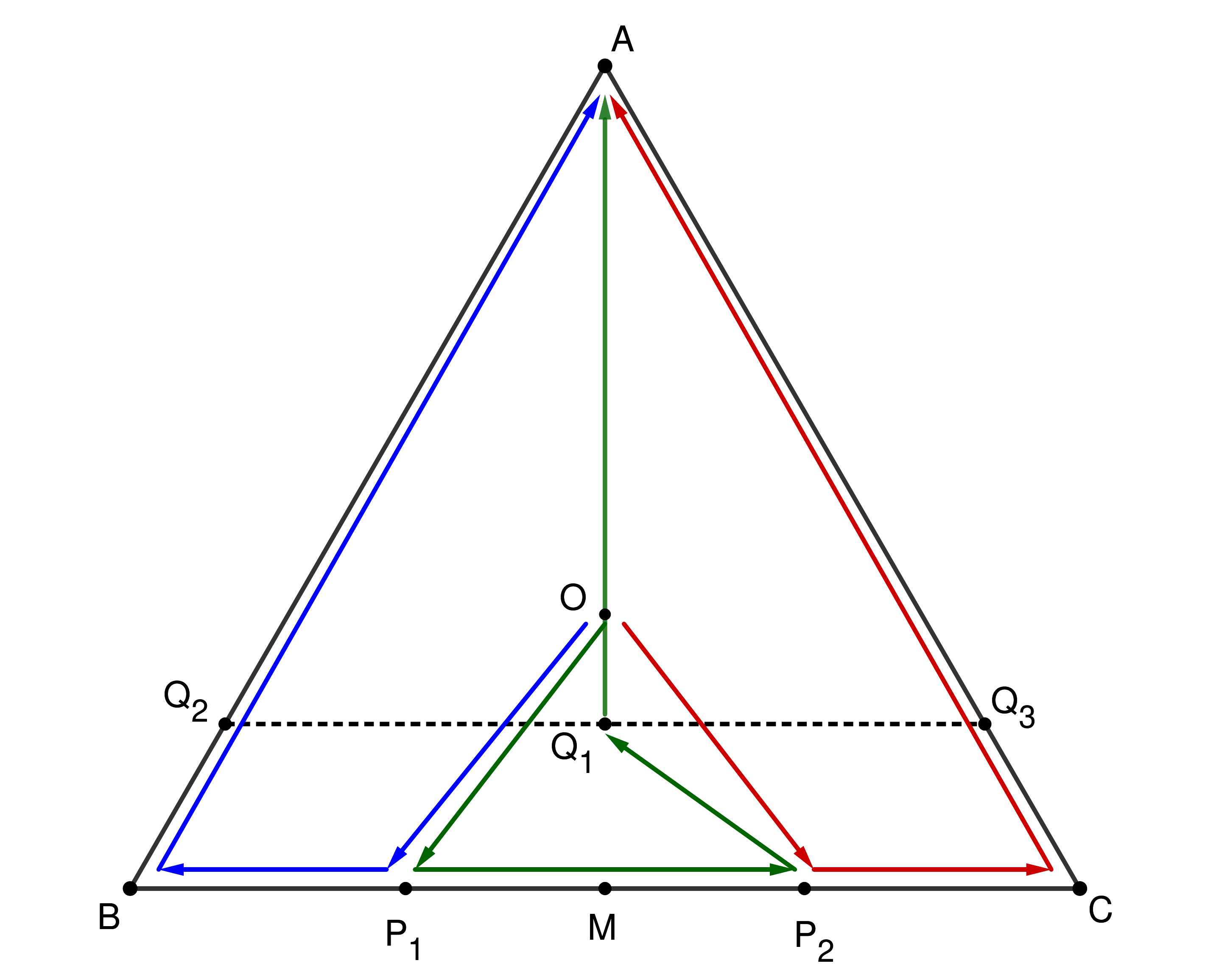}
	\caption{\XoneC trajectories for 3 agents}
	\label{fig:3r-semi-wireless}
\end{subfigure} 
\caption{}
\end{figure}
The {\em Explore 3 sides before Connecting (X3C)} Trajectories of the three agents
are  defined in Trajectories \ref{traj:3r-asym} and shown in Figure \ref{fig:r3asym}. 
We partition the perimeter of $T$ into 4 segments. Three of the segments are assigned to individual agents for exploration. After the exploration of these segments are finished, they move inside $T$ to {\em distributed meeting points} $J_1$, $J_2$ and $J_3$. These points have the following properties:
\begin{itemize}
	\item They are at distance $r$ from each other.
	\item Their distances to point $O$ are equal.
	\item Points $J_1$, $J_2$ and $J_3$ are located on line segments $OM_2$, $OM_3$ and $OM_1$ respectively.
\end{itemize}

After the information is exchanged at the distributed meeting points, if the exit is not found, they all move toward the fourth segment. 
\begin{trajectory}\label{traj:3r-asym} Three agents: \\
\hspace*{2cm}	$R_1: <O,P_1,B,Q_1,J_1,P_3 \mbox{ and wait for }R_3,P_2>$\\
\hspace*{2cm}	$R_2: <O,Q_1,A,Q_2,J_2,P_3 \mbox { and wait for }R_3,P_2>$\\
\hspace*{2cm}	$R_3: <O,Q_2,C,P_2,J_3,P_1,P_2>$
\end{trajectory}

From the above it follows that $|J-1C|=|J_2B|=|J_3A|$. At this point, due to the difference between the distance of each agent to point $P_1$, they don't move together. Only $R_3$ moves toward point $P_1$ and both $R_1$ and $R_2$ move toward $P_3$, the midpoint of segment $P_1P_2$ and wait there for $R_3$. If $R_3$ has found the exit, they move back toward point $P_1$ and if not, they move toward point $P_2$ together. It is obvious that $R_1$ and $R_2$ moving to $P_3$ does not have any negative effect on the worst case evacuation time, since if the exit is close to point $P_1$, agent $R_3$ from $P_1$ has to travel $\frac{|P_1P_2|}{2}$ to inform the other two agents and it takes another $\frac{|P_1P_2|}{2}$ for them to get to the exit, and if the exit is located near $P_2$, it again takes $|P_1P_2|$ for $R_3$ to get to the exit from $P_1$. 

We design the trajectories so that agents arrive at the distributed meeting points at the same time. Therefore we have:
\begin{enumerate}
	\item $t_1=|OP_1|+|P_1B|+|BQ_1|+|Q_1J_1|$
	\item $t_2=|OQ_1|+|Q_1A|+|AQ_2|+|Q_2J_2|$
	\item $t_3=|OQ_2|+|Q_2C|+|CP_2|+|P_2J_3|$
\end{enumerate}
On the other hand, at the end of the first phase when information is exchanged, there will be two critical points:
(1)	 for $R_2$ to reach point $B$, and
(2)	 for $R_3$ to finish the unexplored part of the triangle.
Putting the constraints together, we obtain the following equations:
\begin{enumerate}
	\item $t_1=t_2=t_3$ and
	\item $|J_2B|=|J_3P_1|+|P_1P_2|$
\end{enumerate}
Solving these equations with Maple software, we achieve the results.
Notice that if $r=0$, algorithm \XthreeC converges to Equal Travel Early Meeting 
algorithm in \cite{ChuangpishitMNO17} and 
our result for $r=0$ is identical to their results. 
This algorithm has the lowest evacuation time for $r=0.1761$,  and from then on
the total evacuation time starts to increase. 

\subsection{Explore 1 Side before Connecting (X1C)}\label{sec:4-2}
As $r$ increases, the evacuation time in the \XthreeC algorithm starts to get larger. The \XoneC strategy yields a lower evacuation time than \XthreeC for larger values of $r$. 
The trajectories of  agents are illustrated in Figure \ref{fig:3r-semi-wireless} and defined in Trajectories \ref{traj:3r-semi-wireless}.

\begin{trajectory}\label{traj:3r-semi-wireless}
$\ $ Three agents: \\
 \hspace*{2cm}	$R_1 \mbox{ follows the trajectory:} <O,P_1,B,A>$,\\
 \hspace*{2cm}	$R_2 \mbox{ follows the trajectory:} <O,P_1,P_2,Q_1,A>$,\\
 \hspace*{2cm}	$R_3 \mbox{ follows the trajectory:} <O,P_2,C,A>$,\\
 \hspace*{1cm}where $P_1$ and $P_2$ are located at the same distance from $M$.
\end{trajectory}

The locations of  points $Q_2$, $Q_1$ and $Q_3$ are selected so that
agents $R_1$, $R_2$ and $R_3$ reach them at the same time, say $t$, 
and after that, they will be in effective communication range with each other. 
\\
The location of point $Q_1$ depends on the value of $r$:

{\bf Case 1:} ($0 \leq r <0.5$) Point $Q_1$ is the midpoint of segment $Q_2Q_3$, where $Q_2$ and $Q_3$ are chosen so that $|AQ_2|=|AQ_3|=2r$. Since at time $t$ agents $R_1$ and $R_2$ should be at points $Q_2$ and $Q_1$ respectively, we have
	$|OP_1|+|P_1B|+|BQ_2|=|OP_1|+|P_1P_2|+|P_2Q_1|$.
	Using $|P_2Q_1|=\sqrt{|Q_1M|^2+(|P_1P_2|/{2})^2}$ and $|Q_1M|=3y-|AQ_1|$, we get
$(1-|P_1P_2|)/{2}+(1-2r)=|P_1P_2|+ $ \linebreak $ \sqrt{(3y-\sqrt{3r^2-4r+1})^2+(|P_1P_2|/{2})^2}$ 
from which $|P_1P_2|$ can be obtained as a function of $r$.
Evacuation time for this case is $|OP_1|+|P_1B|+|BQ_2|+|Q_2C|$.\\
{\bf Case 2:} ($0.5 \leq r <\nicefrac{2}{3}$) Point $Q_1$ is positioned on segment $MP_2$ such that 
	$|BP_1|=|P_1P_2|+|P_2Q_1| \mbox{ and } |BQ_1|=r$
	By solving these two equations we get $P_1P_2=r/2$. The evacuation time is $|OP_1|+|P_1B|+|BC|=\sqrt{y^2+({r}/{4})^2}+(\nicefrac{1}{2}-\nicefrac{r}{4})+1$.\\
{\bf Case 3:} ($r \geq \nicefrac{2}{3}$) Point $Q_1$ is at distance $\nicefrac{2}{3}$ from $B$. Then we have $|P_1P_2|=\nicefrac{1}{3}$ and the evacuation time is $|OP_1|+|P_1B|+|BC|=\sqrt{y^2+({1}/{6})^2}+{1}/{3}+1$.


As seen from Table \ref{tab:3r-results}, algorithm \XoneC has better evacuation time 
than \XthreeC for $r>0.22589$, and for $r\geq 0.7$ it achieves the same evacuation  
time as the algorithm in \cite{EtrSq} that uses  $r=1$. 

\subsection{Evacuation of Four Agents}
\label{sec:4-agents}
 \XthreeC and \XoneC can be generalized for 4 agents, with the difference that  the communication range for small value of $r$ can now be  extended to $2r$ in \XthreeC algorithm employing two agents as relays. 
\begin{figure}[t]
\begin{subfigure}{0.5\textwidth}
	\centering
	\includegraphics[scale=0.22]{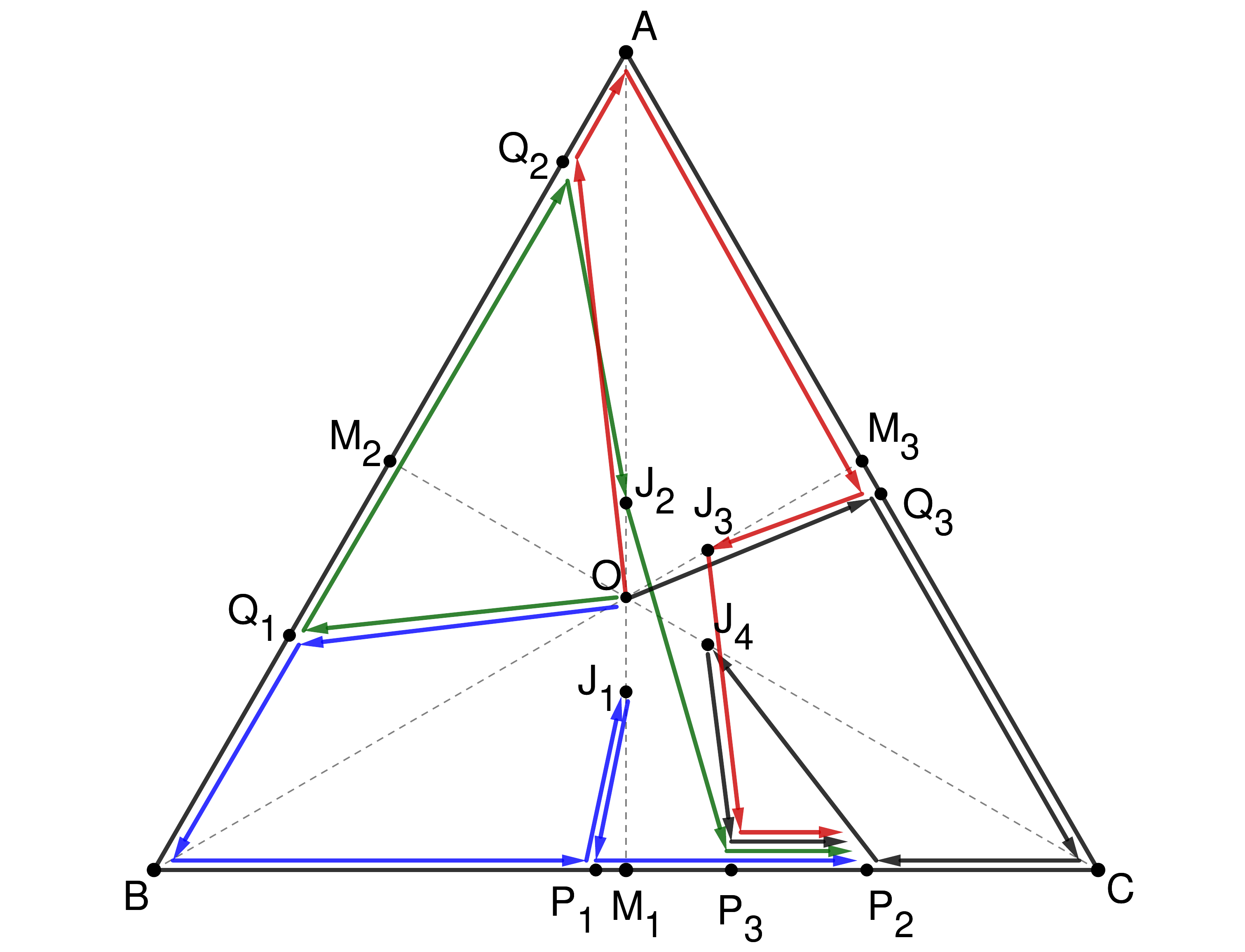}
	\caption{Illustration of \XthreeC algorithm for 4 agents}
	\label{fig:4asym}
\end{subfigure}
\begin{subfigure}{0.5\textwidth}
	\centering
	\includegraphics[scale=0.22]{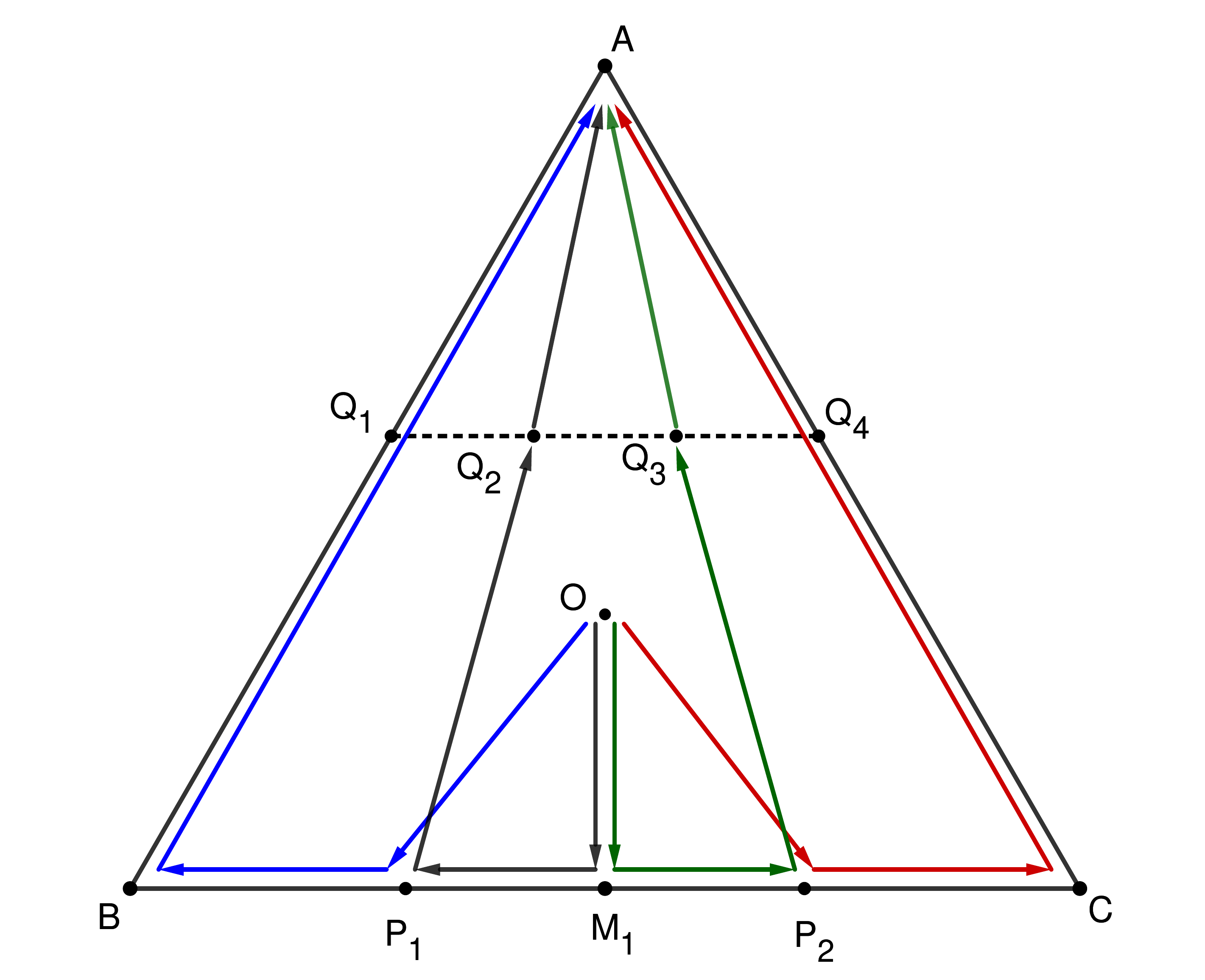}
	\caption{Illustration of \XoneC algorithm for 4 agents}
	\label{fig:4-semi-wireless}
\end{subfigure} 
\caption{}
\end{figure}
\begin{table}[!h]
\begin{center}
\begin{subtable}{0.45\textwidth}
	\begin{center}
		\begin{tabular}{ |c|c|c| } 
			\hline
		r  & \XthreeC &\XoneC  \\ \hline 
		0	&	2.08872		&	2.64971\\
			0.1 & 	2.07849		&	2.37052\\
			0.2 &	2.07642		&	2.13056\\
		0.22589&	2.07714		&	2.07572\\
			0.25 &	2.07828		&	2.02747\\
			0.3 & 	2.08210		&	1.93620\\
			0.4 & 	2.09689		&	1.78880\\
			0.5 & 	2.13037	 	&	1.68958\\
			0.6 & 	N/A	 		&1.67532\\
			$\geq 0.7$  & 	N/A	 		&1.666667\\
			\hline
		\end{tabular}
	\end{center}
\caption{ }
	\label{tab:3r-results}
\end{subtable}
\hfill
\begin{subtable}{0.45\textwidth}
	\begin{center}
		\begin{tabular}{ |c|c|c| } 
			\hline
			$r$ & 	\XthreeC & 		\XoneC \\ \hline 
			0.0 & 	1.98157 & 	2.59944\\
			0.1 & 	1.96199 & 	2.19408\\
			0.11619	&	1.95993 & 2.13688\\
			0.1721 & 1.95993 & 1.95993 \\
			0.2 & 	1.95993 & 	1.88392\\
			0.3 & 	N/A		 & 	1.67649\\ 
			0.4 & 	N/A		 & 	1.62573\\
			0.5 & 	N/A		 & 	1.61912\\ 
			0.6 & 	N/A		 & 	1.61302\\ 
			0.7 & 	N/A		 & 	1.61050\\ 
			1.0 & 	N/A		 & 	1.61050\\ 
			\hline
		\end{tabular}
	\end{center}
\caption{ }
	\label{tab:r4}	
\end{subtable}	
\end{center}
	\caption{A comparison of evacuation times of Algorithms \XthreeC and \XoneC for (a) three agents and (b) four agents.}
\end{table}

\subsection{\XthreeC Trajectories for 4 agents}
We describe the \XthreeC algorithm for 4 agents here, see Figure
 ~\ref{fig:4asym}. The triangle is divided into 5 segments, each agent is responsible for exploring one of these segments and after the exploration of these 4 segments is finished, all the agents will move to the inside of the triangle in order to exchange information. In the case that the exit is not found, the agent closest to the fifth segment, in our case $R_1$, will start exploring that segment and the other 3 agents will move to the midpoint of fifth segment and will wait for the other agent. The trajectories of all 4 agents are shown in Figure \ref{fig:4asym}. 

%


\begin{trajectory}\label{traj:4r-asym}
$\ $ Four agents\\
\hspace*{0.5cm} $R_1 \mbox{ follows the trajectory:} <O,Q_1,B,P_1,J_1,P_1,P_2>$\\
\hspace*{0.5cm} $R_2 \mbox{ follows the trajectory:} <O,Q_1,Q_2,J_2,P_3,\mbox{ wait for }R_1,P_2>$\\
\hspace*{0.5cm}	$R_3 \mbox{ follows the trajectory:} <O,Q_2,A,Q_3,J_3,P_3,\mbox{ wait for }R_1,P_2>$\\
\hspace*{0.5cm}	$R_4 \mbox{ follows the trajectory:} <O,Q_3,C,P_2,J_4,P_3,\mbox{ wait for }R_1,P_2>$
\end{trajectory}
Same as Algorithm \XthreeC the length of trajectories traversed by each agent up to the the point where they reach the distributed meeting points, i.e. points $J_{1-4}$, are equal. Also points $P_1$ and $P_2$ are chosen in a way to satisfy the equation $|J_1P_1|+|P_1P_2|=|J_1C|$.

In this algorithm it is very critical for $R_2$ to reach point $P_3$ before $R_1$ does, as $R_2$ is the farthest agent to $P_3$. And this is true for small values of $r$, but as $r$ gets larger than 0.11619, $|J_2P_3|$ gets larger than $|J_1P_1|+|P_1P_3|$ and hence, the evacuation time of the algorithm starts to increase at this point.

\subsection{\XoneC trajectories for 4 agents}
The difference with the \XoneC algorithm for 3 agents us that the trajectories of agents are symmetric with regards to line $AO$, see Figure  \ref{fig:4-semi-wireless}.  At first agents travel to different points on edge $BC$ and start exploring that edge. After the exploration of edge $BC$ is finished, two of the agents continue exploring the other two edges while the two remaining agents will move inside the triangle in order to maintain a communication link between two agents moving alongside the edges. The trajectories of agents are illustrated in Figure \ref{fig:4-semi-wireless}.


\begin{trajectory} \label{traj:4r-semi-wireless}
$\ $ Four agents\\
\hspace*{2cm}	$R_1 \mbox{ follows the trajectory:} <O,P_1,B,A>$\\
\hspace*{2cm}	$R_2 \mbox{ follows the trajectory:} <O,M_1,P_1,Q_2,A>$\\
\hspace*{2cm}	$R_3 \mbox{ follows the trajectory:} <O,M_1,P_2,Q_3,A>$\\
\hspace*{2cm}	$R_4 \mbox{ follows the trajectory:} <O,P_2,C,A>$
\end{trajectory}

Note that if $r$ is small and the exit is located at vertex $C$, the agents have to travel a long way to the top of $T$ so that they could communicate and then they have to travel all the way down to the exit point. Hence this algorithm is not efficient for small values of $r$. 
\\
Since the trajectories are symmetric, we only analyze the trajectories of $R_1$ and $R_2$ and assume the exit is located at vertex $C$. Based on the value of $r$ we have the following cases for the position of points $Q_1$ and $Q_2$:
\begin{description}
	\item[Case1, $r<\nicefrac{1}{3}$:] Point $Q_1$ is at distance $3r$ from $A$ on edge $AB$. Point $Q_2$ is $r$-apart from point $Q_1$ on segment $Q_1Q_4$. Point $P_1$ is chosen in a way to satisfy the following equation:
	 $$|OM_1|+|M_1P_1|+|P_1Q_2|=|OP_1|+|P_1B|+|BQ_1|$$
	 
	 \item[Case2, $\nicefrac{1}{3}\leq r <0.6436493404$:] In this case, since $|P_1P_2|>r$, when $R_2$ and $R_3$ reach points $P_1$ and $P_2$, they have to move toward each other to be in communication range. Points $Q_2$ and $Q_3$ are placed on the edge $BC$, and points $P_1$ and $Q_2$ are obtained from solving the following two equations:
	 \begin{itemize}
		\item $y+|M_1P_1|+|P_1Q_2|=|OP_1|+|P_1B|$ and
		\item $|M_1P_1|-|P_1Q_1|=\nicefrac{r}{2}$ 
	\end{itemize}	  
	 The evacuation time is $|OP_1|+|P_1B|+1$.
	 \item[Case3, $r\geq 0.6436493404$:] With large enough $r$, this problem converges to the wireless problem proposed in \cite{CKKNOS}. Points $P_1$ and $P_2$ are chosen in a way so that all the 4 agents finish exploring side $BC$ at the same time, i.e. $y+|M_1P_1|=|OP_1|+|P_1B|$. When the edge $BC$ is explored all the agents move toward vertex $A$. At all time of the execution of the algorithm, all the agents are in effective communication range. The evacuation time will be $y+|M_1P_1|+1 \approx 1.610499805$.
\end{description}

Note that when $R_2$ and $R_3$ move toward vertex $A$, they move with speed less that 1 in a way that they always be on the line going through $R_1$ and $R_4$. The comparison of evacuation times of \XthreeC and \XoneC for 4 agents are shown in Table \ref{tab:r4}


%
See Table \ref{tab:r4} for evacuation times of  \XthreeC and \XoneC for four agents.

\section{Evacuation of $k>4$ agents}\label{sec:ka}

It is shown in \cite{EtrSq} that $1+2y\approx 1.5773$ is a lower bound on the 
evacuation time of
 $k$ wireless agents (i.e., $r=1)$ from the centroid of the triangle, for any number $k$ of agents. It follows that this is also a lower bound for any $0\leq r \leq 1$. It is shown in \cite{EtrSq} that this lower bound time can be achieved with 6 agents with $r=1$. We show below that for any  $0<r<1$, evacuation can be done in time $1+2y$ using the \CXP strategy, however, the minimum number of agents needed is inversely proportional to $r$.   

First, we show that for $0<r<1$, the lower bound on the  evacuation time of $1+2y$ cannot be achieved with a constant number of agents.

\begin{figure}{r}
	\centering
 \includegraphics[scale=0.45]{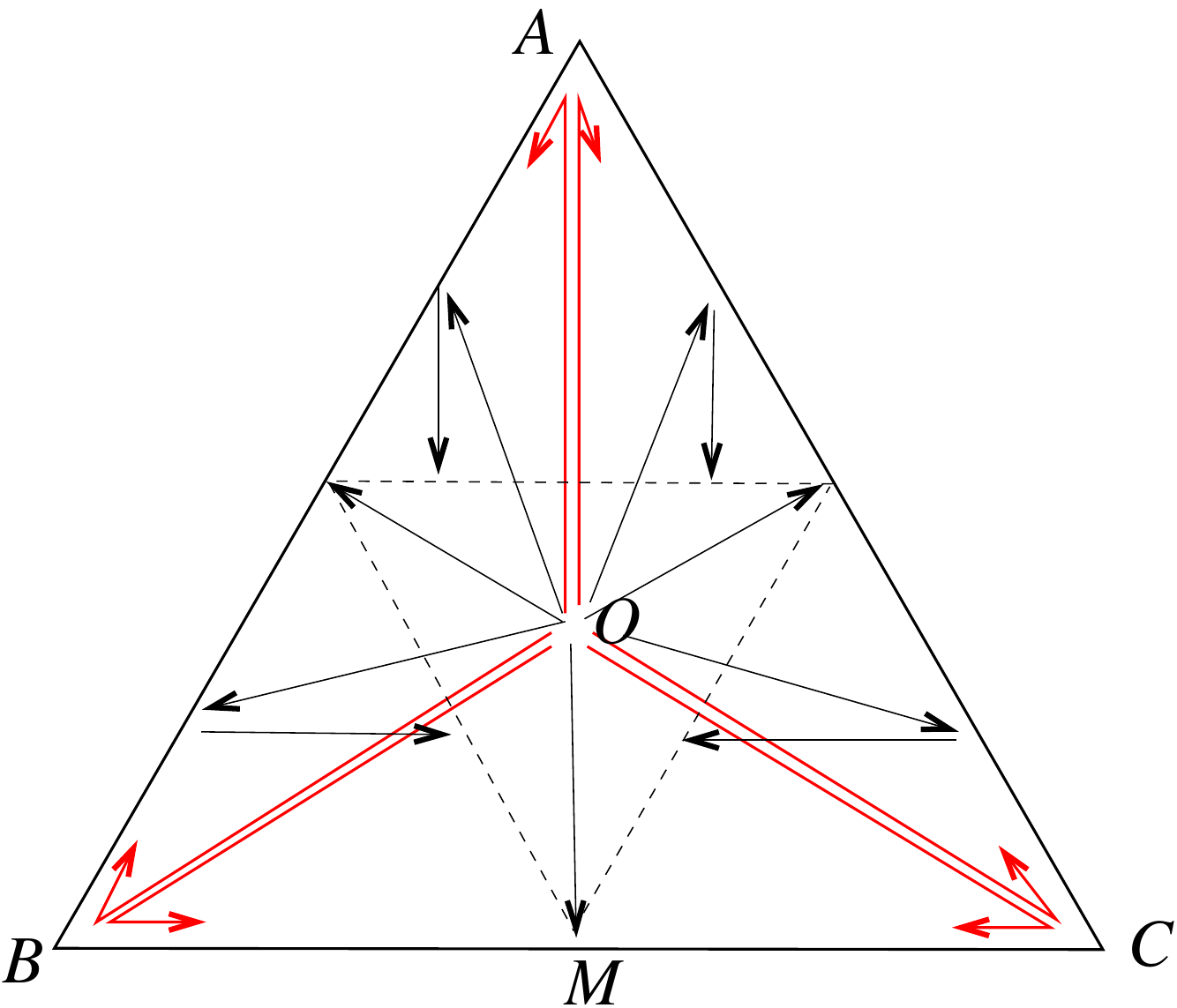}
 \caption{Trajectories of $k$ agents  for $k=12$, where 6 agents (red) are exploring, and the remaining agents form a relay network.}
 \label{subfig:kagents}
        \vspace*{-5mm}
\end{figure}

\begin{theorem} \label{th:krobots}
  Given transmission range $r$, the number of 
agents needed to achieve the optimal evacuation time $1+2y$ is 
at least $1/r + 1$.
\end{theorem}
 \begin{proof} 
Let $t$ be the time the first agent, say $R_1$, reaches a vertex, say $A$. Clearly,
$t\geq 2y$. Since the adversary can place the exit  at either $B$ or $C$, for the evacuation time to be exactly $1+2y$, it must be that $t=2y$, and furthermore, another agent must have reached either $B$ or $C$ or both, and must be able to instantly communicate the presence of the exit to $R_1$. For this communication to happen, an additional $1/r-1$ agents are needed, for a total of $1/r+1$ agents.
\end{proof} 

Next we show that for any $0<r<1$, the \CXP strategy can achieve this lower bound with a sufficient number of agents. 
\begin{theorem} \label{th:krobots2}
For any $0 < r < 1$, the evacuation of  $k=6 +2\lceil(\frac{1}{r}-1)\rceil$ agents of transmission range $r$ from the centroid of an equilateral triangle can be done in time $1+2y\approx  1.5773$, which is optimal. 
\end{theorem}


\begin{proof}
 Let $i=  \lceil \frac{1}{r}-1\rceil$. The trajectories of  the agents are as shown in Figure \ref{subfig:kagents}. Each vertex is reached by two agents and they explore the perimeter of the triangle from that vertex until the mid-point on each edge. Notice that the exploration terminates at time $2y+0.5$, and if an agent finds the exit at time $2y+t$, other exploration agents are at distance at most $1-t$ from it.

Furthermore, $2i$ agents go to  edges $AB$ and $AC$ into equidistant positions to form a relay network for the 6 agents doing the exploration. 
 When an exploring agent reaches a relay agent, the relay agent starts to move to its final position on the interior  dashed triangle. In this way  the relay agents can perform the relay function for the exploring agents, and are also  able to reach the exit when it is found within the  bound $1+2y$.
\end{proof}
\bibliographystyle{plain}
\bibliography{refs}
\newpage

\end{document}